\documentclass{JHEP3}
\usepackage{amsmath}
\usepackage{amssymb}
\usepackage{amsthm}
\usepackage{cite}
\usepackage{graphicx}
\usepackage{epsfig}
\usepackage{verbatim}
\usepackage{epstopdf}

\DeclareMathOperator{\area}{area}
\DeclareMathOperator{\Tr}{Tr}
\theoremstyle{plain}
\newtheorem{thm}{Theorem}[section]
\newtheorem{lem}[thm]{Lemma}
\newtheorem{prop}[thm]{Proposition}
\newtheorem{cor}[thm]{Corollary}
\theoremstyle{definition}

\title{Properly ordered dimers, $R$-charges, and an efficient inverse algorithm}
\author{Daniel R. Gulotta \\ Department of Physics, Princeton University \\
Princeton, NJ 08544, USA \\ \email{dgulotta@princeton.edu}}
\preprint{PUPT-2273}
\abstract{The $\mathcal{N}=1$ superconformal field theories that
arise in AdS-CFT from
placing a stack of D3-branes at the singularity of a toric Calabi-Yau
threefold can be described succinctly by dimer models.  We present
an efficient algorithm for constructing a dimer model from the geometry of
the Calabi-Yau.  Since not all dimers produce consistent field theories, we
perform several consistency checks on the field theories
produced by our algorithm: they have the correct number of gauge groups,
their cubic anomalies agree with the Chern-Simons coefficients in the
AdS dual, and all gauge invariant chiral operators satisfy the unitarity bound.
We also give bounds on
the ratio of the central charge of the theory to the area of the toric diagram.
To prove these results, we introduce the concept of a properly ordered dimer.}
\begin{document}
\section{Introduction}
The AdS-CFT correspondence \cite{Maldacena:1997re, Gubser:1998bc, Witten:1998qj}
tells us that Type IIB string theory on
$AdS_5 \times X_5$, where $X$ is a five-dimensional 
Sasaki-Einstein manifold, is dual to a
four-dimensional $\mathcal{N}=1$ superconformal gauge theory.
We can study the gauge theory by placing D3-branes at a singularity of
$Y_6$, the cone over $X_5$, which is a Calabi-Yau threefold.

In the case where $Y_6$ is toric,
dimer models \cite{Kasteleyn-1967,Hanany:2005ve,Franco:2005rj, Franco:2005sm, Hanany:2005ss, Feng:2005gw, stienstra-2007}
are a convenient way of encoding the field content and
superpotential of the CFT.
One can try to compute the geometry from the dimer or vice versa.
There are algorithms for solving the former problem by taking the determinant
of the Kasteleyn matrix \cite{Kasteleyn-1967,Hanany:2005ve,Franco:2005rj,Franco:2005sm,Hanany:2005ss, Feng:2005gw, stienstra-2007} and by counting the windings of zigzag paths \cite{Hanany:2005ss, Feng:2005gw, stienstra-2007}.
The latter problem can be solved by the
``Fast Inverse Algorithm'' \cite{Hanany:2005ss, Feng:2005gw, stienstra-2007},
although the algorithm is computationally infeasible for all but very
simple toric varieties due to the large amount of trial and error required.
We resolve this problem by eliminating the need for trial and error.
Our algorithm uses some ideas from the Fast Inverse Algorithm
and the method of partial resolution of the toric singularity
\cite{Douglas:1997de, Morrison:1998cs, Feng:2000mi, GarciaEtxebarria:2006aq}.

One difficulty in constructing dimers is that not every dimer describes
a consistent field theory.  One way of determining that a field theory
is not consistent is by counting its faces.  Each face represents a gauge
group, and a consistent theory should have as many gauge groups as there are
cycles for Type IIB D-branes to wrap in the AdS theory.
Previously there was not a simple, easy to check criterion for determining that
a dimer is consistent.  We propose that any dimer that has the correct number
of faces and that has no nodes of valence one is consistent.
We will present several pieces of evidence to support our proposal.

If the dimer is consistent, then the cubic anomalies of the CFT should be
equal to the Chern-Simons coefficients
of the AdS dual \cite{Witten:1998qj, Benvenuti:2006xg}.  We show that equality
holds in dimers that meet our two criteria.

In a four-dimensional SCFT the unitarity bound says that
each gauge invariant scalar operator should have dimension at least
one \cite{Mack-1977}, and the $R$-charge of a chiral
primary operator is two-thirds
of its dimension \cite{Dobrev-1985}.  However, when we try to
compute the $R$-charge of a gauge
invariant chiral primary operator in an inconsistent dimer theory, the answer is
sometimes less than two-thirds.  We will show that in dimers that meet our two
criteria, the $R$-charges of chiral primary operators are always at least
two-thirds if the number of colors is sufficiently large.

We also show that dimers that meet our two criteria have the properties
that corner perfect matchings are unique, and that the
zigzag path windings agree with the $(p,q)$-legs of the toric diagram.

While studying $R$-charges, we prove that
$\frac{27N^2 K}{8 \pi^2} < a \le \frac{N^2 K}{2}$ for toric theories,
where $a$ is the cubic 't Hooft anomaly $\frac{3}{32}(3 \Tr R^3 - \Tr R)$,
$N$ is the number of colors of each gauge group, and $K$ is the area of the
toric diagram (which is half the number of gauge groups).

\section{\label{sec:defs}Definitions}

A \emph{dimer model} \cite{Kasteleyn-1967,Hanany:2005ve,Franco:2005rj, Franco:2005sm, Hanany:2005ss, Feng:2005gw, stienstra-2007}
consists of a graph whose vertices are colored
black or white, and every edge connects a white vertex to a black vertex,
i.~e.\ the graph is bipartite.  We will use dimer models embedded on the torus $T^2$ to
describe toric quiver gauge theories.

A \emph{perfect matching} of the dimer is a set of edges of the dimer such that
each vertex is an endpoint of exactly one of the edges.  The difference of two
perfect matchings is the set of edges that belong to exactly one of the
matchings.

The \emph{Kasteleyn matrix}
is a weighted adjacency matrix of the dimer.
There is one row for each white vertex and one column for each black
vertex.  Let $\gamma_w$ and $\gamma_z$ be a pair of curves whose winding
numbers generate the homology group $H^1(T^2)$.  The weight of an edge
is $c w^a z^b$ where $c$ is an arbitrary nonzero complex number\footnote{
The original definition of the Kasteleyn matrix imposes constraints
on $c$ for the purpose of counting perfect matchings \cite{
Hanany:2005ve,Franco:2005rj, Franco:2005sm, Hanany:2005ss}.  However, these
constraints are not necessary for determining the Newton polygon.
We follow the convention of \cite{Feng:2005gw}, which points out that
it is useful for the purposes of mirror symmetry to allow arbitrary
nonzero coefficients.},
$w$ and $z$ are variables,
$a$ is the number of times $\gamma_w$ crosses the edge
with the white edge endpoint on its left minus the number of times
$\gamma_w$ crosses the edge with the white endpoint on its right and
$b$ is defined similarly with $\gamma_w$ replaced by $\gamma_z$.
The determinant of this matrix tells us the geometry of the field
configuration.

The \emph{Newton polygon} of a multivariate polynomial is the convex hull of the
set of exponents of monomials appearing in the polynomial.
The Newton polygon of the determinant is known as the
\emph{toric diagram}.  If we choose a different basis for computing the
Kasteleyn matrix, then the toric diagram changes by an affine transformation.

A \emph{$(p,q)$-leg} of a toric diagram is a line segment drawn perpendicular
to and proportional in length to a segment joining consecutive boundary
lattice points of the diagram.

A \emph{zigzag path} is a path of the dimer on which edges alternate between
being clockwise adjacent around a vertex and being counterclockwise adjacent
around a vertex.  A zigzag path is uniquely determined by a choice of
an edge and whether to turn clockwise or counterclockwise to find the next
edge.  Therefore each edge belongs to two zigzag paths.  (These paths could
turn out to be the same, although it will turn out that we want to work with
models in which they are always different.)

In \cite{Hanany:2005ss} it is conjectured that in a consistent field
theory, the toric diagram can also be computed by looking at the windings of the
zigzag paths: they are in one-to-one correspondence with the $(p,q)$-legs.
The conjecture was proved using mirror symmetry in \cite{Feng:2005gw}.

The \emph{unsigned crossing number} of a pair of closed paths on the torus
is the number of times they intersect.  The \emph{signed crossing number}
of a pair of oriented closed paths on the torus is the number of times they
intersect with a positive orientation (the tangent vector to the second path
is counterclockwise from the tangent to the first at the point of intersection)
minus the number of times they intersect with a negative orientation.  It
is a basic fact from homology theory that the signed crossing number of
a path with winding $(a,b)$ and a path with winding $(c,d)$ is $
(a,b)\wedge(c,d)=ad-bc$.

We will work with the zigzag path diagrams of \cite{Hanany:2005ss}
(referred to there as rhombus loop diagrams).
We obtain a zigzag path diagram from a dimer as follows.  For each edge of the
dimer we draw a vertex of the zigzag path diagram at a point on that edge.
To avoid confusion between the vertices of this diagram and the vertices
of the dimer we will call the latter nodes.
We connect two vertices of the zigzag path diagram if the dimer edges they
represent are consecutive along a zigzag path. (This is equivalent to them
being consecutive around a node and also to them being consecutive around a
face.)  We orient the edges of the zigzag path diagram as follows.  If the
endpoints
lie on dimer edges that meet at a white (resp. black) node, then the edge
should go counterclockwise (resp. clockwise) as seen from that node.
With this definition, each node of the dimer becomes a face of the zigzag path
diagram, with all edges oriented counterclockwise for a white node, or
clockwise for a black node.  The other faces of the zigzag
path diagram correspond to faces of the dimer, and the orientations of their
edges alternate.  Figure \ref{fig:exampledimer} shows an example
of a dimer and its corresponding zigzag path diagram.

Conversely, we can obtain a dimer from a zigzag path diagram provided that
the orientations of the intersections alternate along each path.
Around each vertex of such a zigzag path diagram, there is one face with
all counterclockwise oriented edges, one face with all clockwise oriented
edges, and two faces whose edge orientations alternate.  Draw a white
node at each counterclockwise oriented face and a black node
at each clockwise oriented face, and connect nodes whose faces share a
corner.

\section{\label{sec:proper}Consistency of dimer field theories}

\subsection{Criteria for consistency and inconsistency}

One difficulty in dealing with dimer models is that not all of them
produce valid field theories.  While there are a number of ways of 
determining that a dimer produces an invalid field theory
there has not yet been a simple criterion for showing
that a dimer theory is valid.

One way of proving that a dimer produces an invalid field theory is
by counting the number of faces of the dimer, i.~e.\ the number of gauge
groups.  If the dimer theory is
consistent, then the number of gauge groups should equal
the number of 0, 2, and 4-cycles in the Calabi-Yau around which
D3, D5, and D7-branes, respectively, can wrap \cite{Benvenuti:2004dy}.
The Euler characteristic of the Calabi-Yau is the number of even dimensional
cycles minus the number of odd dimensional cycles.
There are no odd dimensional cycles, so the number of gauge groups should be
equal to the Euler characteristic.
The Euler characteristic of a toric variety
equals twice the area of the toric diagram \cite{Fulton-1993}.

We propose that a dimer will produce a valid field theory if
the dimer has no nodes of valence one and
it has a number of faces equal to twice the area of the lattice
polygon whose $(p,q)$-legs are the winding numbers of the zigzag paths.
(Recall that this polygon is the same as the Newton polygon of the determinant
of the Kasteleyn matrix for consistent theories.)  In this section, we will show
that dimers satisfying our two criteria also have the properties that
their cubic anomalies agree with the Chern-Simons coefficients of
the AdS dual, the
$R$-charges of gauge invariant chiral primary
operators are greater than or equal to two-thirds, the windings of the
zigzag paths are in one-to-one correspondence
with the $(p,q)$-legs of the toric diagram, and the corner perfect matchings
are unique.

It will be convenient to
introduce a property that we call ``proper ordering'', which will
turn out to be equivalent to the property of having the correct number of
faces and no valence one nodes.
We call a node of the dimer \emph{properly ordered} if the order of the
zigzag paths around that node is the same as the circular order of the
directions of their windings.  (We do not allow two zigzag paths with
the same winding to intersect, nor do we allow zigzag paths of winding zero,
since these scenarios make the ordering ambiguous.)
We call a dimer properly ordered if
each of its nodes is properly ordered.

\begin{thm}
  A connected dimer is properly ordered iff it has no valence one nodes
  and it has a number of faces equal to twice the area of the 
  convex polygon whose $(p,q)$-legs
  are the winding numbers of the zigzag paths of the dimer.
\end{thm}

\begin{proof}
  A properly ordered dimer cannot have a valence one node, since such a node
  would be the endpoint of an edge that is an intersection of a zigzag path
  with itself.
  Therefore it suffices to prove that a dimer with no valence one nodes
  is properly ordered iff it has a number of faces equal to twice the area of
  the convex polygon whose $(p,q)$-legs
  are the winding numbers of the zigzag paths of the dimer.

  Define the ``winding excess'' of a node $v$ of the dimer as follows.
  Let $\mathbf{w}_0, \mathbf{w}_1, ..., \mathbf{w}_{n-1}$ be the
  winding numbers of the zigzag paths passing through $v$
  (in the order that the paths appear around $v$).
  Start at $\mathbf{w}_0$ and turn counterclockwise to $\mathbf{w}_1$, then
  counterclockwise to $\mathbf{w}_2$, etc., and finally counterclockwise
  back to $\mathbf{w}_0$.  Then the winding excess is defined as the number
  of revolutions that we have made minus one.
  (In the special case where $\mathbf{w}_i$ and $\mathbf{w}_{i+1}$ are
  equal or one of them is zero, we count one-half of a revolution.)
  A node is properly ordered iff it has winding excess zero and none of the
  $\mathbf{w}_i$ are zero and no two consecutive windings are equal.
  A node with a $\mathbf{w}_i=0$ or $\mathbf{w}_i = \mathbf{w}_{i+1}$
  can have winding excess zero only if it has exactly two edges
  (and hence two zigzag paths passing through it).
  There must be some other node
  that is an endpoint of one of the edges where the two zigzag paths intersect,
  and that has more than two edges (since the graph is connected).
  This node cannot have winding excess zero.
  So all nodes are properly ordered iff all nodes have winding excess zero.
  A node has negative winding excess iff it has just one edge, and we have
  assumed that the dimer has no such nodes.
  Therefore the dimer is properly ordered iff the sum of all
  of the winding excesses is zero.

  If we choose a node and draw all of the wedges between
  the consecutive winding numbers, then the winding excess is the number
  of wedges containing any given ray minus one.
  (We can think of the special case
  of consecutive winding numbers being the same as the average of
  a full wedge and an empty wedge, and the case of a zero winding number as
  the average of wedges of all angles.)
  Now consider the sum of the winding excess over all vertices.
  A pair of oppositely oriented intersections between two zigzag paths
  forms two full wedges and therefore contributes two to the sum.  The sum of
  the contributions from
  unpaired intersections can be computed as follows.  Label the winding
  numbers $\mathbf{w}_0, \mathbf{w}_1, ..., \mathbf{w}_{n-1}$, ordered by
  counterclockwise angle from some ray $\mathcal{R}$.  (A zigzag path with
  zero winding number has no unpaired intersections, so it is not included.)
  Then for $i<j$ the
  unpaired wedges formed by $\mathbf{w}_i$ and $\mathbf{w}_j$ will contain
  $\mathcal{R}$ iff $\mathbf{w}_i \wedge \mathbf{w}_j < 0$.  There are
  $2 |\mathbf{w}_i \wedge \mathbf{w}_j|$ unpaired wedges
  ($|\mathbf{w}_i \wedge \mathbf{w}_j|$ unpaired crossings of the zigzag
  paths, and each appears in two vertices).
  So the number of unpaired wedges formed by
  $\mathbf{w}_i$ and $\mathbf{w}_j$ containing $\mathcal{R}$ equals
  $\max(- 2 \mathbf{w}_i \wedge \mathbf{w}_j, 0)=|\mathbf{w}_i \wedge \mathbf{w}_j|-\mathbf{w}_i \wedge \mathbf{w}_j$.
  Since $\sum_{i<j} |\mathbf{w}_i \wedge \mathbf{w}_j|$ is the number of
  unpaired edges, it follows that the number of wedges containing $\mathcal{R}$
  is the number of paired edges plus the number of unpaired edges minus
  $\sum_{i<j} \mathbf{w}_i \wedge \mathbf{w}_j$, or
  $E-\sum_{i<j} \mathbf{w}_i \wedge \mathbf{w}_j$, where $E$ is the total
  number of edges of the dimer.  The sum of the
  winding excesses is $E-V-\sum_{i<j} \mathbf{w}_i \wedge \mathbf{w}_j=
  F-\sum_{i<j} \mathbf{w}_i \wedge \mathbf{w}_j$,
  where $V$ and $F$ are the number of nodes and faces of the dimer,
  respectively.  We have
  $\sum_{i<j} \mathbf{w}_i \wedge \mathbf{w}_j = \sum_{i} \mathbf{w}_i \wedge
  \sum_{j>i} \mathbf{w}_j$.  If we lay the winding vectors tip-to-tail, then
  $\mathbf{w}_i \wedge \sum_{j>i} \mathbf{w}_j$ is twice the area of the triangle
  formed by the tail of $\mathbf{w}_0$ and the tip and tail of $\mathbf{w}_i$.
  Hence $\sum_{i} \mathbf{w}_i \wedge \sum_{j>i} \mathbf{w}_j$ is twice the
  area of the convex polygon formed by all the winding vectors.  If we rotate
  the polygon 90 degrees then we get a polygon whose $(p,q)$-legs are the
  winding numbers.  So the
  sum of the winding deficiencies of the nodes is zero iff $F$ equals twice
  the area of the lattice polygon whose $(p,q)$-legs are the zigzag path winding
  numbers.
\end{proof}

\subsection{Some perfect matchings of properly ordered dimers}

\begin{FIGURE} {
    \label{fig:cornermatching}
    \centerline{
      \epsfig{figure=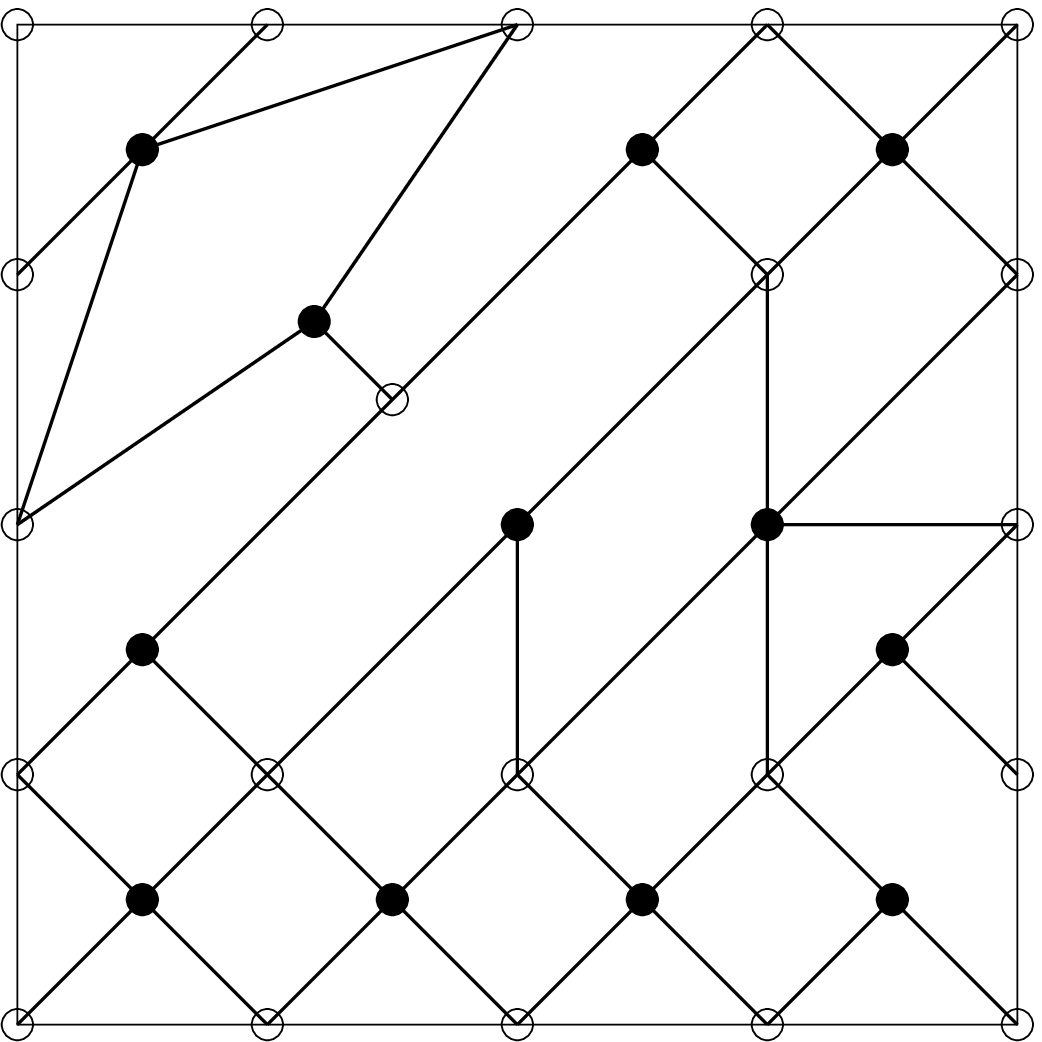,width=1.25in}
      \quad
      \epsfig{figure=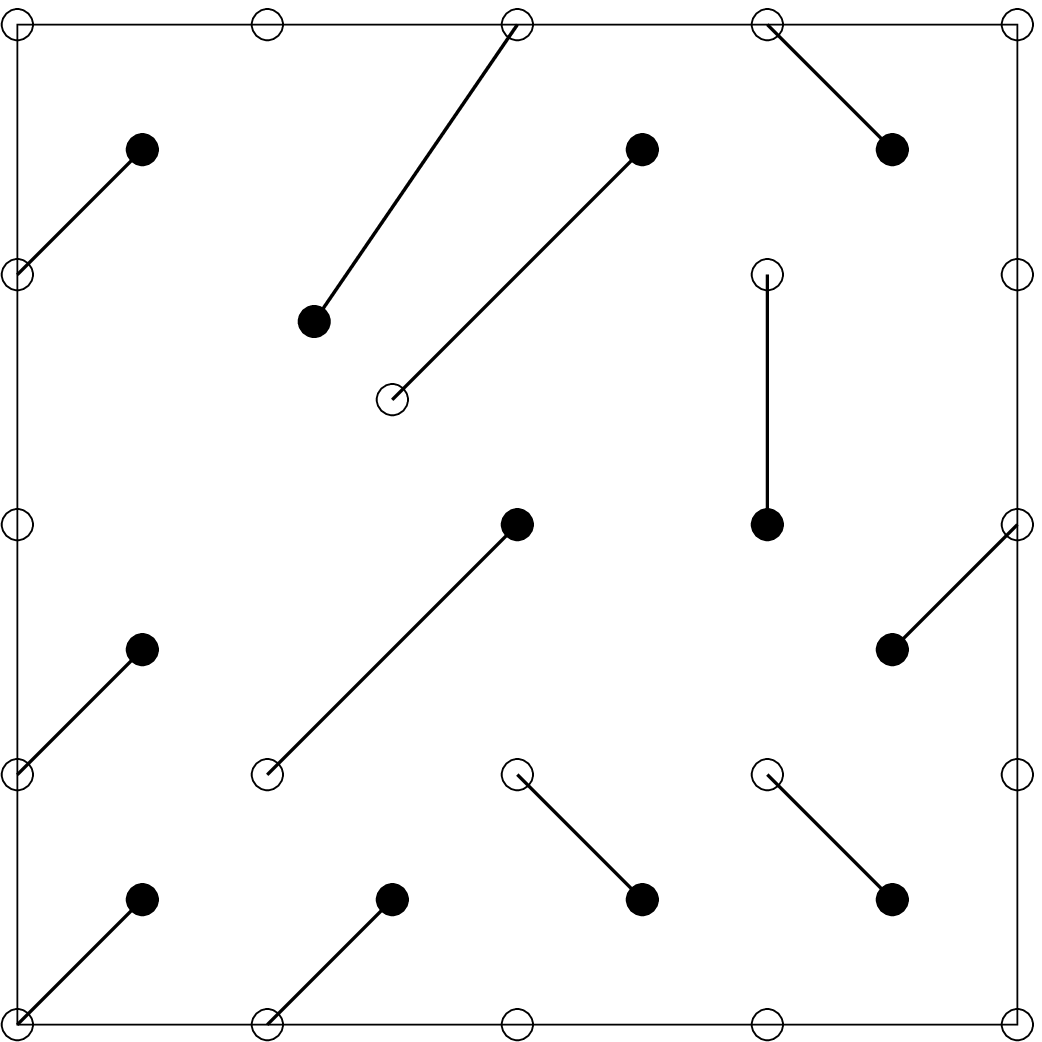, width=1.25in}
    }
    \centerline{}
    \centerline{
      \epsfig{figure=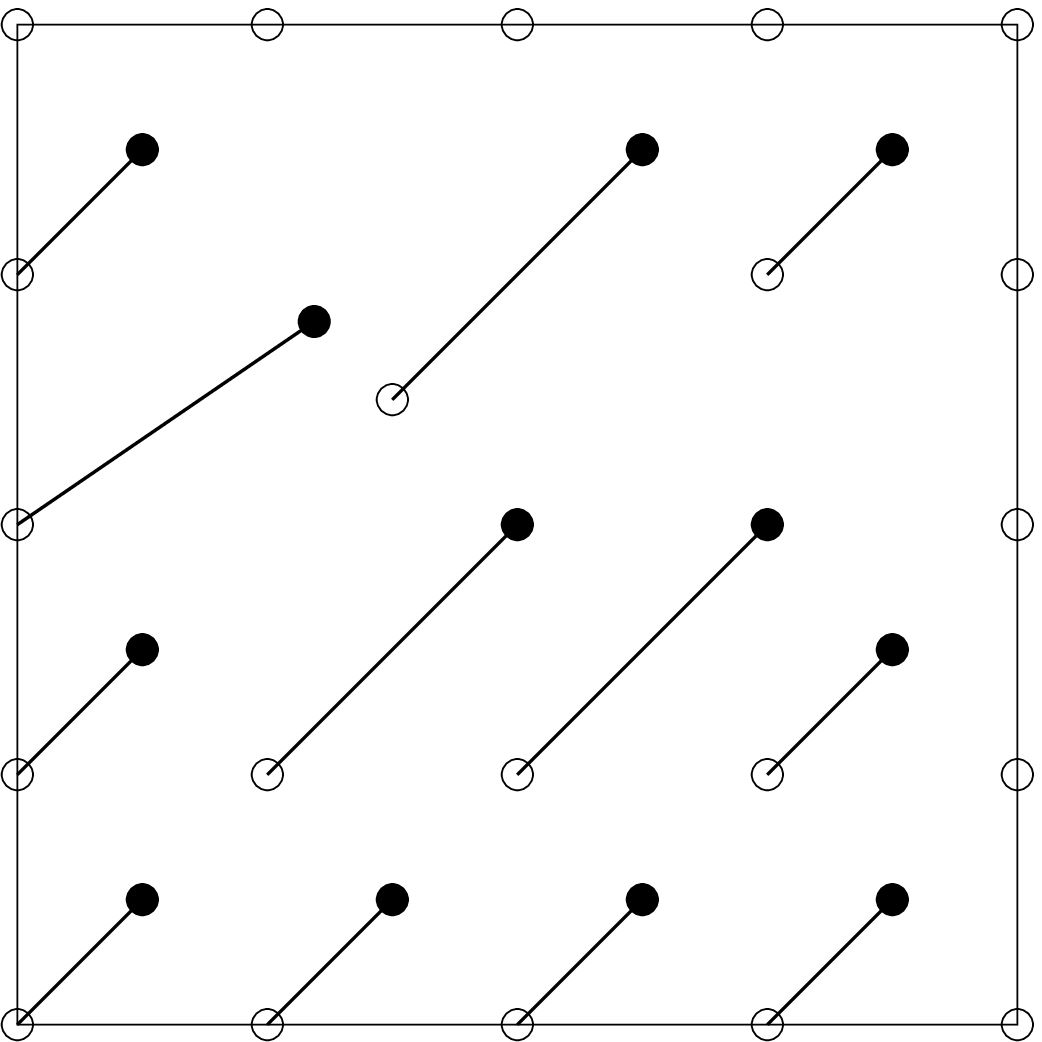,width=1.25in}
      \quad
      \epsfig{figure=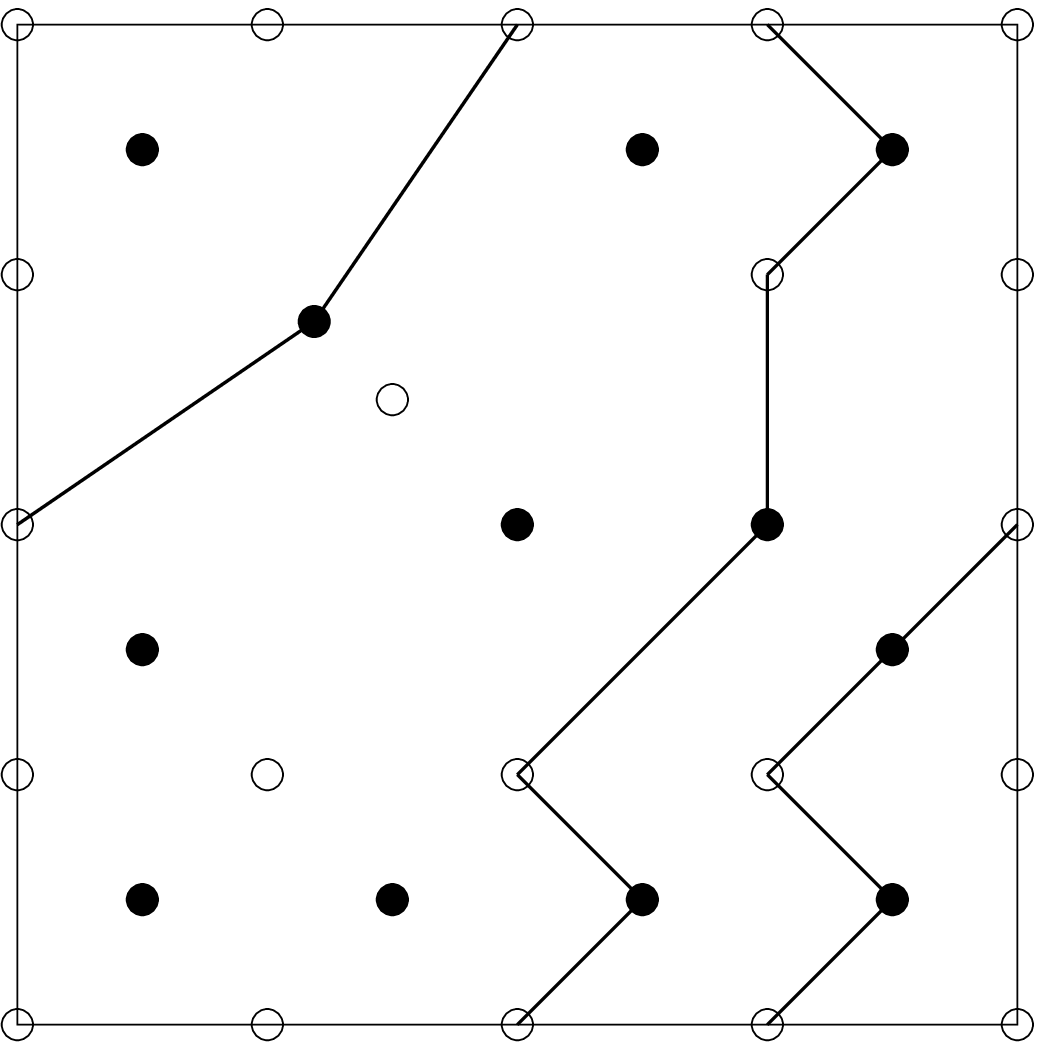, width=1.25in}
    }
    \caption{A dimer, two of its corner perfect
      matchings, and their difference, which is a zigzag path.}
}
\end{FIGURE}

\begin{FIGURE}
{\label{fig:rays}
\centerline {
  \epsfig{figure=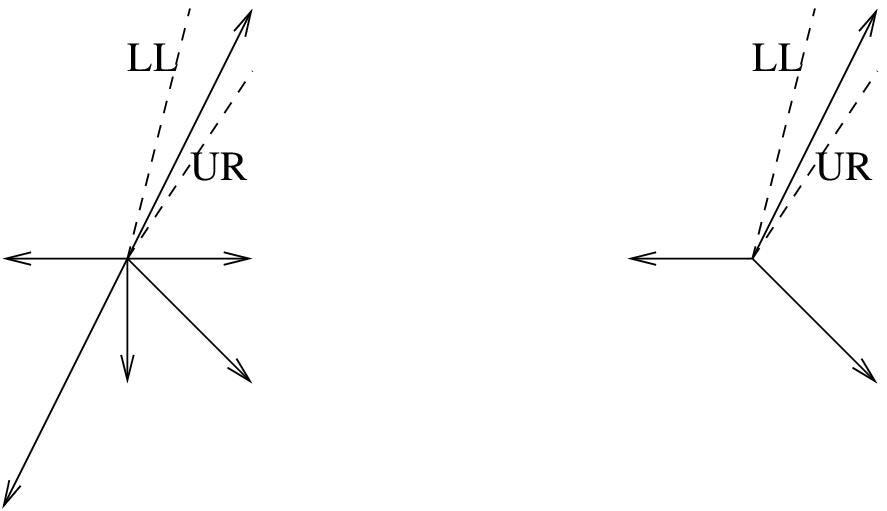, height=1.25in}
}
\caption{Left: The windings of the zigzag paths of the dimer in figure
\ref{fig:cornermatching}.  The dotted lines labeled UR and LL are rays that
yield perfect matchings shown in the upper right and lower left quadrants of 
figure \ref{fig:cornermatching}, respectively.
Right: The windings of the paths passing through the bottom right
black node.
For any node and any edge
ending at that node, the proper ordering criterion implies that the two zigzag
paths to which the edge belongs have adjacent winding directions.  Therefore
in the right diagram, there is a natural correspondence between
edges passing through the node and wedges formed by consecutive arrows.
When constructing a perfect matching $M(\mathcal{R})$,
we choose the wedge containing $\mathcal{R}$.  In the left diagram, there
is a one-to-one correspondence between wedges and corners of the toric
diagram.}
} \end{FIGURE}

We will construct some perfect matchings that will turn out to
correspond to the corners of the toric diagram.
Our construction of the perfect matchings
is similar to Theorem 7.2 of \cite{stienstra-2007}.
Let $\mathcal{R}$ be any ray whose direction does not coincide with
that of the winding number of any zigzag path.  For any node $v$,
consider the zigzag paths passing through $v$ whose winding numbers make
the smallest clockwise and smallest counterclockwise angles with $\mathcal{R}$.
(These paths are unique because the proper ordering condition requires that
all paths through $v$ have different winding numbers.)
By proper ordering, these two zigzag paths
must be consecutive around $v$.  Therefore they share an edge that has $v$ as an
endpoint.  Call this edge $e(v)$.  Let $v'$ be the other endpoint of $e(v)$.
The same two zigzag paths must be consecutive about $v'$ since they
form the edge $e$.  Since $v'$ is properly ordered it must then be the case
that those two paths make the smallest clockwise and
counterclockwise angles with
$\mathcal{R}$ among all paths passing through $v'$.  Hence $e(v)=e(v')$.
So the pairing of $v$ with $v'$ is a perfect matching.  We will call
this matching $M(\mathcal{R})$.  Figure \ref{fig:rays} depicts the
relationship between rays and perfect
matchings.

The following characterization of the boundary perfect matchings
containing a given edge follows immediately from our definition and will
be useful later.
\begin{lem} \label{lem:edgeinpm}
For any edge $e$ of the dimer, let $Z_r$ and $Z_s$ be the zigzag paths
such that $e$ is a positively oriented intersection of $Z_r$ with $Z_s$.
(Equivalently, $e$ is a negatively oriented intersection of $Z_s$ with $Z_r$.)
Let $\mathbf{w}_r$ and $\mathbf{w}_s$ be the windings of $Z_r$ and $Z_s$,
respectively.  Let $\mathcal{R}$ be a ray.  Then $e$ is in $M(\mathcal{R})$
iff $\mathcal{R}$ is in the wedge that goes counterclockwise from
$\mathbf{w}_r$ to $\mathbf{w}_s$.

In particular each edge is in at least one corner perfect matching.
\end{lem}

\subsection{\label{sec:zzp}Zigzag paths and $(p,q)$-legs}

As we mentioned in Section \ref{sec:defs}, it is known
\cite{Hanany:2005ss,Feng:2005gw} that dimers that produce a consistent field
theory have the property that the $(p,q)$-legs of the toric diagram are in
one-to-one correspondence with the winding numbers of the zigzag paths.

\begin{thm} \label{thm:zzp}
  In a dimer with properly ordered nodes, the zigzag paths are in
  one-to-one correspondence with the $(p,q)$-legs of the toric diagram.
\end{thm}

Our proof of Theorem \ref{thm:zzp}
resembles that of Theorem 9.3 of \cite{stienstra-2007}.

\begin{FIGURE} {\label{fig:gammaz}
\centerline{
  \epsfig{figure=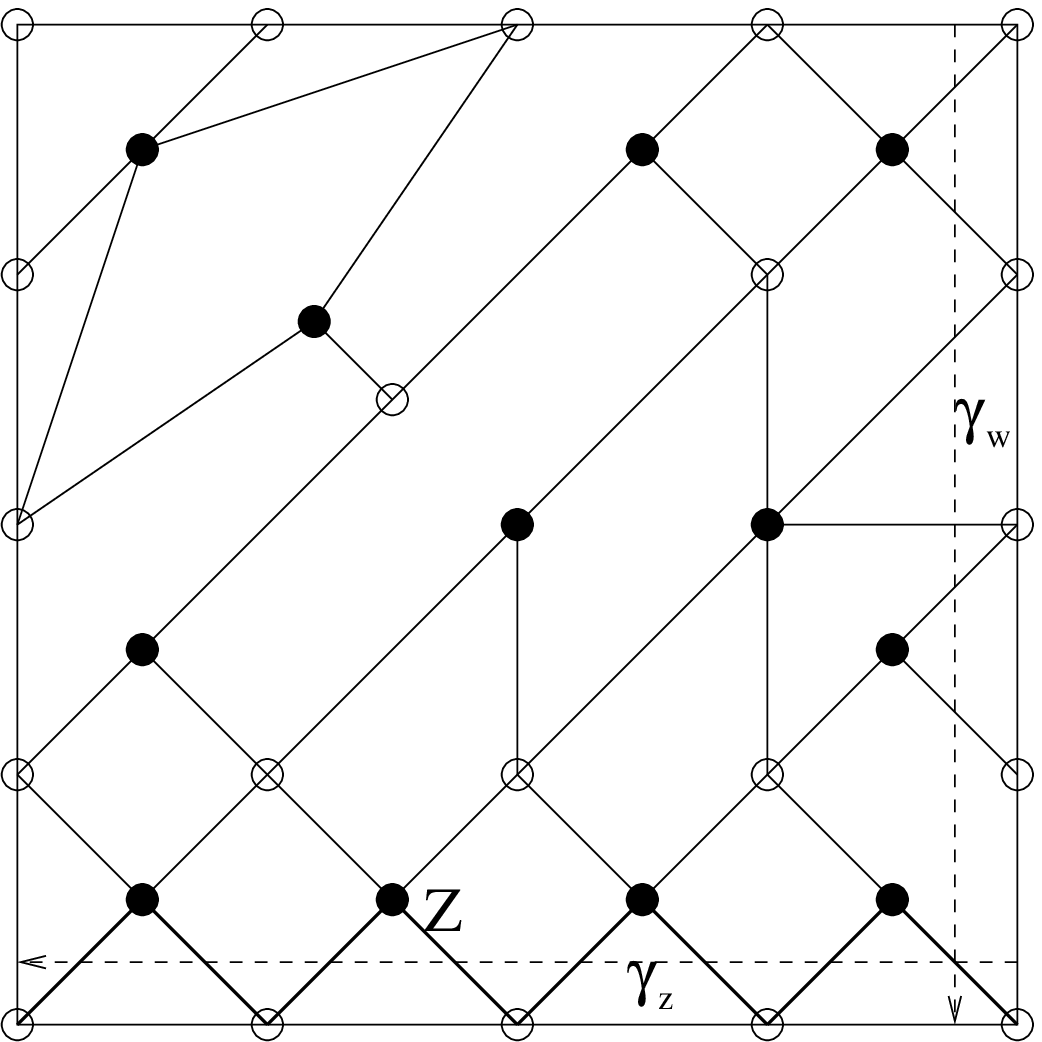,width=2.5in}
}
\caption{The path $\gamma_z$ intersects each edge of the
zigzag path $Z$ and no other edges.  We may choose any path $\gamma_w$ that
completes the basis.}
} \end{FIGURE}

\begin{lem} \label{lem:coordfunc}
For any zigzag path $Z$ in any dimer, the number of intersections of a perfect
matching with $Z$ is a degree one polynomial function of its coordinates.
\end{lem}
\begin{proof}
In computing the Kasteleyn matrix we can choose the path $\gamma_z$ to follow
$Z$, so that the number of times $\gamma_z$ intersects a perfect matching $M$ is
just the number of edges that $M$ and $Z$ have in common.
(See figure \ref{fig:gammaz}.)
For this choice of $\gamma_z$, the point corresponding to $M$
has $y$-coordinate equal to $|M \cap Z|$.
For a different choice of $\gamma_z$, the coordinates differ by an affine
transformation.
\end{proof}

\begin{lem} \label{lem:bdyiszover2}
Let $Z$ be a zigzag path of a properly ordered dimer, and let
$\mathcal{R}_1$ and $\mathcal{R}_2$ be rays such that the winding direction of
$Z$ lies between them and all of the other winding directions do not.
Then there exists a boundary line of the toric diagram passing through
$M(\mathcal{R}_1)$ and $M(\mathcal{R}_2)$ such that all perfect
matchings on this line intersect
$Z$ exactly $\frac{|Z|}{2}$ times, and all perfect matchings not on the line
intersect $Z$ fewer than $\frac{|Z|}{2}$ times.
\end{lem}
\begin{proof}
Since the winding number of $Z$ is adjacent to $\mathcal{R}_1$,
$M(\mathcal{R}_1)$ must choose one of the two $Z$-edges of each node
that has them.  Hence $|M(\mathcal{R}_1) \cap Z|=\frac{|Z|}{2}$ and similarly
$|M(\mathcal{R}_2) \cap Z|=\frac{|Z|}{2}$.
No perfect matching can contain more than half of the edges
of the path.  Therefore the toric diagram lies in the half plane
that, in the coordinate system of Lemma \ref{lem:coordfunc}, is given
by the equation $y \le \frac{|Z|}{2}$. 
$M(\mathcal{R}_1)$ and $M(\mathcal{R}_2)$ are both on the boundary.
\end{proof}

\begin{prop} \label{prop:corner}
  The matchings $M(\mathcal{R})$ lie on the corners of the toric diagram.  The
  order of the corners around the boundary is the same as the order of the
  ray directions.
\end{prop}

\begin{proof}
The intersection of all half planes described in the proof of
Lemma \ref{lem:bdyiszover2} is the convex hull
of all of the $M(\mathcal{R})$'s.  Conversely, each $M(\mathcal{R})$ is
in the toric diagram.
So the toric diagram must be
the convex hull of the $M(\mathcal{R})$'s.

Each $M(\mathcal{R})$ must be at a corner of the toric diagram since it is
contained in two different boundary lines (one for the first counterclockwise
zigzag path direction from $\mathcal{R}$ and another for the first clockwise
zigzag path direction).  Furthermore, if $\mathcal{R}_1$ and $\mathcal{R}_2$
have only one winding direction between them, then they share a
boundary line and hence $M(\mathcal{R}_1)$ and $M(\mathcal{R}_2)$ lie
on consecutive corners.
\end{proof}

\begin{proof}[Proof of Theorem \ref{thm:zzp}]

Let $\mathbf{w}$ be the winding number of a zigzag path, and let $n$ be the
number of zigzag paths with that winding.
Let $\mathcal{R}_1$ and $\mathcal{R}_2$ be rays such that $\mathbf{w}$ lies
between them and all other winding directions do not.  By Proposition
\ref{prop:corner}, $M(\mathcal{R}_1)$ and $M(\mathcal{R}_2)$ lie on
consecutive corners of the toric diagram. 
An edge belonging to one of the zigzag paths of winding $\mathbf{w}$ will
be in either $M(\mathcal{R}_1)$ or $M(\mathcal{R}_2)$ but not both, while
all other edges are in neither or both perfect matchings.
Therefore the difference of the two perfect matchings is just the union
of the zigzag paths with winding $\mathbf{w}$.
Therefore the toric diagram points corresponding to
$M(\mathcal{R}_1)$ and $M(\mathcal{R}_2)$ are separated by
$-n \mathbf{w}^{\perp}$, where $-\mathbf{w}^{\perp}$ is the 90 degree clockwise
rotation of $\mathbf{w}$.
This proves the theorem.
\end{proof}

\subsection{\label{sec:unique}Unique corner perfect matchings}

It is generally believed that dimers that have more than one perfect
matching at a corner of the toric diagram are inconsistent
\cite{Hanany:2005ve, Butti:2005vn, Hanany:2005ss}.  We show that properly
ordered dimers have unique corner perfect matchings.

\begin{thm}
If a dimer is properly ordered, then each corner of the toric
diagram has just one perfect matching.
\end{thm}
\begin{proof}
Suppose there exists a perfect matching $M'$ that shares a toric diagram point
with $M(\mathcal{R})$ but is not equal to $M(\mathcal{R})$.
Consider the set of zigzag paths that contain an edge that is in 
$M(\mathcal{R})$ or $M'$ but not both.  Let $Z$ be one with
minimal counterclockwise angle from $\mathcal{R}$.
Let $v$ be a node of the dimer through which $Z$ passes.
If $v$ includes a zigzag path with
winding between $\mathcal{R}$ and that of $Z$, then $M(\mathcal{R})$
and $M'$ are the same at that vertex.  If not, then
$M(\mathcal{R})$ chooses one of the edges of $Z$ at $v$.
Recall that Lemma \ref{lem:coordfunc} says that the number of
intersections with $Z$ depends only on the toric diagram point.
Therefore $M'$ has the same number of edges in $Z$ as $M(\mathcal{R})$.
Since $M(\mathcal{R})$ chooses an edge of $Z$ at every node where
$M(\mathcal{R})$ and $M'$ differ, equality can hold only if
$M'$ chooses the other edge of $Z$ at every such node.
If we start at an edge of $Z$ that is in $M(\mathcal{R})$ but not $M'$
and alternately follow edges of the $M(\mathcal{R})$ and $M'$, then we will
traverse a cycle that lies entirely in $Z$.
Since zigzag paths in properly ordered dimers do not intersect themselves,
the cycle must be $Z$.
Then both $M(\mathcal{R})$ and $M'$ contain half the edges
of $Z$.  So the winding number of $Z$ is either the closest or
farthest from $\mathcal{R}$ in the counterclockwise direction.
If $Z$ were the farthest, then $M(\mathcal{R})$ and $M'$ would have to be the
same because every edge of the dimer would be in at least one zigzag path whose
winding is closer to $\mathcal{R}$ in the counterclockwise direction than
$Z$'s.  So $Z$ must be the closest in the counterclockwise direction.

Now let $Z'$ be a zigzag path with minimal clockwise
angle from $\mathcal{R}$ on which $M(\mathcal{R})$ and $M'$ differ.
By the same reasoning as above, we find the winding direction of $Z'$ is
the closest to $\mathcal{R}$ in the clockwise direction and 
that $M(\mathcal{R})$ and $M'$ have no edges of $Z'$ in common.
Since $Z$ and $Z'$ represent consecutive sides of the toric diagram,
the crossing number of $Z$ and $Z'$ must be nonzero.
A node can have two edges belonging to both $Z$ and $Z'$ only if they have
opposite orientations, i.~e.\ they contribute zero to the signed crossing
number.
Therefore there must be a node with only one $Z$-$Z'$
intersection.  $M'$ must include this edge because it includes an edge of $Z$
and $Z'$ at every node that has one, but it cannot include this edge because
it does not share any edges of $Z$ with $M(\mathcal{R})$.
Therefore our assumption that there existed a matching $M'$ differing from
$M(\mathcal{R})$ but sharing the same toric diagram point must be false.
\end{proof}

Once we know that the corner matchings are unique, we can also classify
all of the boundary perfect matchings.
\begin{cor} \label{cor:bdymatchings}
Consider a point $A$ on the boundary of the toric diagram such that the
nearest corner $B$ in the counterclockwise direction is $p$ segments away
and the nearest corner $C$ in the clockwise direction is $q$ segments away.
Then each perfect matching at $A$ may be obtained from the perfect matching
associated to $B$ by flipping $p$ zigzag paths and from the perfect matching
associated to $C$ by flipping $q$ zigzag paths.  The number of perfect
matchings at $A$ is $\binom{p+q}{q}$.
\end{cor}
\begin{proof}
For any boundary perfect matching $M$
there exists a winding $\mathbf{w}$ such that $M$ contains half the edges
of each zigzag path of winding $\mathbf{w}$.
For any zigzag path $Z$ of winding $\mathbf{w}$, we can delete the half
of the $Z$-edges that are in $M$ and add the other half.  This operation
moves the perfect matching one segment along the boundary of the toric diagram.
There can be at most $p$ zigzag paths for which the operation moves
the toric diagram point counterclockwise and at most $q$ zigzag paths for
which the operation moves the point clockwise.  But there are a total of
$p+q$ zigzag paths of winding $\mathbf{w}$, so there must be exactly $p$
of the former and $q$ of the latter.  Consequently we see that $M$ can
be obtained from a corner perfect matching by flipping $p$ zigzag paths
(or from a different corner perfect matching by flipping $q$ zigzag paths).
The number of ways of choosing the paths to flip is
$\binom{p+q}{p}$.
\end{proof}

\subsection{\label{sec:rcharge}$R$-charges and cubic anomalies}

The $R$-charges of the fields may be determined by \emph{$a$-maximization}
\cite{Intriligator:2003jj}.
First, we impose the constraint that the $R$-charge of each superpotential
term should be two.  We also impose the constraint that the beta function
of each gauge group should be zero.  These conditions can be expressed as
\begin{eqnarray}
\label{eq:rvert} \sum_{e \in v} R(e) & = & 2 \\
\label{eq:rface} \sum_{e \in f} \left[1-R(e)\right] & = & 2.
\end{eqnarray}

Among all $U(1)$ symmetries satisfying these constraints, the $R$-symmetry
is the one that locally maximizes the cubic 't Hooft anomaly
\begin{equation} \label{eq:thooft}
a = \frac{9N^2}{32} \left[F+\sum_e (R(e)-1)^3 \right].
\end{equation}

Butti and Zaffaroni \cite{Butti:2005vn} have proposed some techniques for
simplifying the computation of the $R$-charge.  For any perfect matching $M$
we can define a function $\delta_M$ that takes the value $2$ on all edges
in the perfect matching and zero on all other edges.  Any such $\delta_M$
automatically satisfies (\ref{eq:rvert}).
Butti and Zaffaroni noted that in some cases the perfect matchings on
the boundary of the toric diagram yield functions that also satisfy
(\ref{eq:rface}), and these functions span the set of solutions to 
(\ref{eq:rvert}) and (\ref{eq:rface}).  We will show that
their observation is true for properly ordered dimers.  
\begin{thm} \label{thm:bdysymmetry}
In a dimer with properly oriented nodes, the solutions to (\ref{eq:rvert})
and (\ref{eq:rface}) are precisely the linear combinations of $\delta_M$,
for boundary perfect matchings $M$.
\end{thm}

\begin{FIGURE}
{ 
  \centerline{
    \epsfig{figure=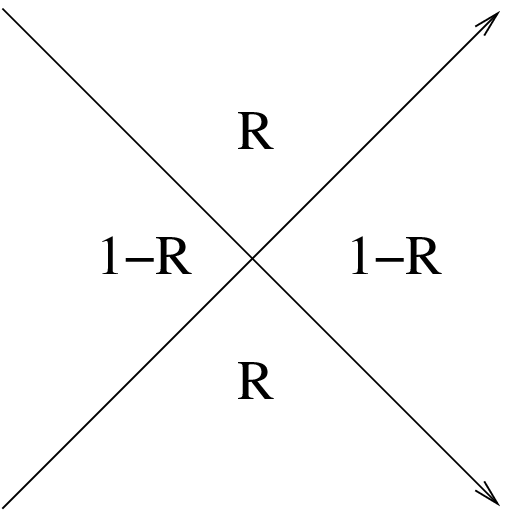, width=1.5in}
  }
\caption{\label{fig:rcontribution}The contribution of the vertex to the
equations (\ref{eq:rcomb}) for the four surrounding faces.}
}
\end{FIGURE}

\begin{FIGURE}
{ 
  \centerline{
    \epsfig{figure=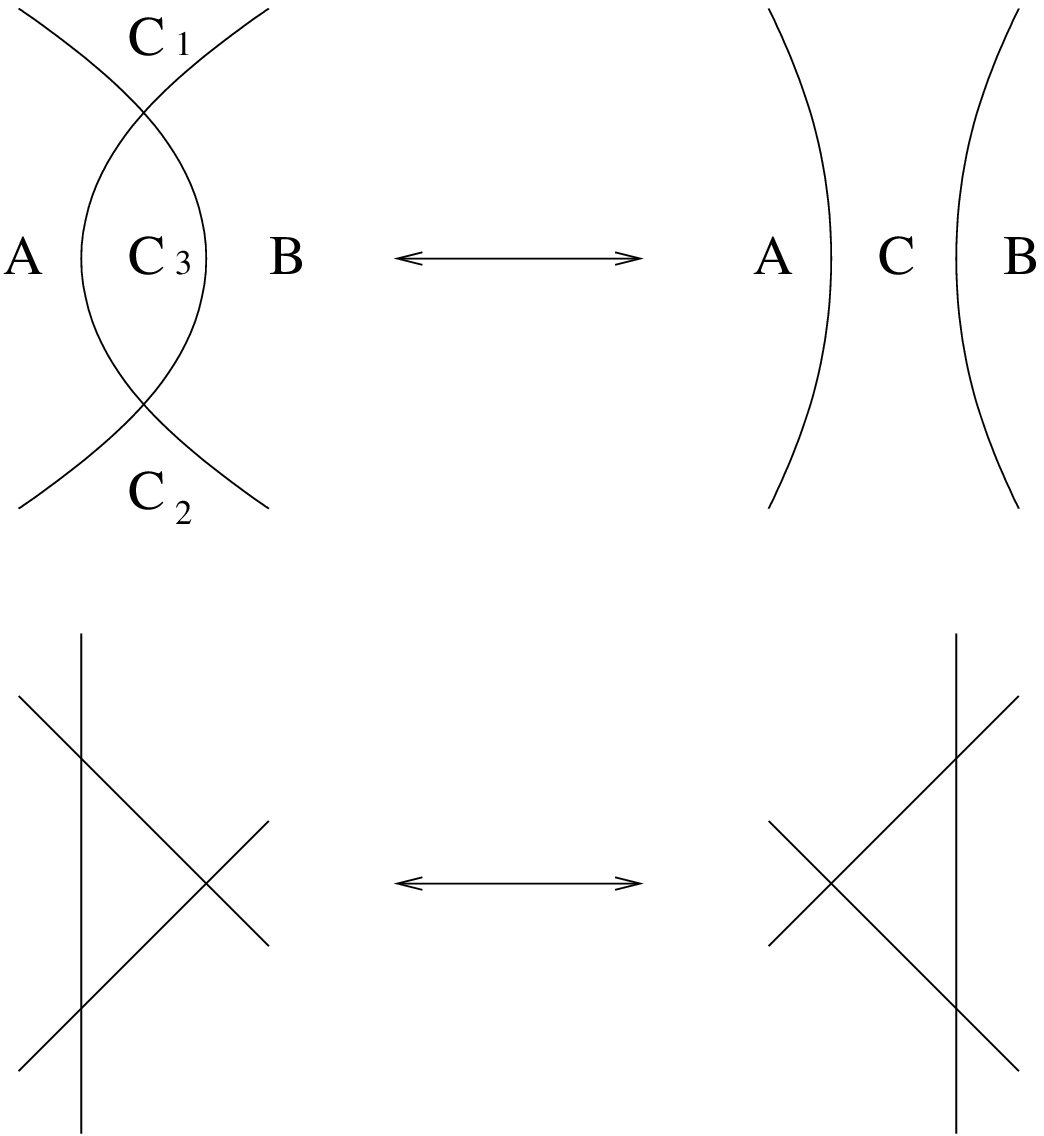, width=3.5in}
  }
\caption{\label{fig:rearrange}Combinatorial changes in the zigzag path diagram as a result of
deformations.}
}
\end{FIGURE}

We first determine the dimension of the solution space of
(\ref{eq:rvert}) and (\ref{eq:rface}), so that we will be able to show that
there are not any more solutions beyond the boundary $\delta_M$.

\begin{lem} \label{lem:pathsminus1}
For any dimer in which the zigzag paths have winding numbers that are
prime (i.~e.\ their $x$ and $y$ components are relatively prime, or equivalently,
they can each be sent to $(1,0)$ by an $SL_2(\mathbb{Z})$ transformation)
and not all parallel and in which no zigzag path intersects itself,
the set of solutions to (\ref{eq:rvert}) and (\ref{eq:rface}) has dimension
equal to the number of zigzag paths minus one.
\end{lem}
\begin{proof}
First we will show that the number of solutions depends only on the winding
numbers of the zigzag paths.  We will work with the zigzag path diagram.
In this diagram, $R$ is a function on vertices.  We can unify 
(\ref{eq:rvert}) and (\ref{eq:rface}) into a single equation as follows.
We first define the function $\sigma_{v,f}(x)$, where $v$ is a vertex of the
zigzag path diagram and $f$ is a face of the zigzag path diagram having $v$ as a
corner.  If the two zigzag paths containing $v$ are similarly oriented around
$f$, then $\sigma_{v,f}(x)=x$; if they are oppositely oriented around $f$ then
$\sigma_{v,f}(x)=1-x$.  Then (\ref{eq:rvert}) and (\ref{eq:rface}) can be
expressed as
\begin{equation} \label{eq:rcomb}
\sum_{v \in f} \sigma_{v,f}(R(v))=2.
\end{equation}
(See Figure \ref{fig:rcontribution}).

We can deform any zigzag path diagram with non-self-intersecting zigzag
paths to any other zigzag path diagram with non-self-intersecting zigzag
paths with the same winding numbers.  As the diagram is deformed, it can
change combinatorially in several ways: a pair of intersections between
a pair of zigzag paths can be added or removed, or a zigzag path can be
moved past the crossing of two other zigzag paths.  Figure
\ref{fig:rearrange} illustrates these possibilities.  Note that at intermediate
steps, the zigzag path diagram may not correspond to a dimer, but we can
still consider the set of solutions to (\ref{eq:rcomb}). 

First consider the case where a pair of intersections between a pair of zigzag
paths is added or removed.
If $C_1$ is not the same face as $C_2$, then the values of the two new crossings
are constrained by the equations for $C_1$ and $C_2$ and the dimension of the
set of solutions to (\ref{eq:rvert}) and (\ref{eq:rface}) remains unchanged.
If the two zigzag paths have winding numbers that are not parallel, then they
must intersect somewhere else, which implies $C_1 \ne C_2$.  If the winding
numbers are parallel, then there must be some other zigzag path
whose winding number is not parallel to either and hence must intersect both.
Again $C_1 \ne C_2$.

Now consider the case where a zigzag path is moved past the crossing of two
other zigzag paths.  We can check that any solution to (\ref{eq:rcomb}) in
the first diagram is also a solution to (\ref{eq:rcomb}) in the second
diagram, and vice versa.  So performing the move depicted in the second diagram
does not change the solution set.   So we have shown that the dimension of
the solution space to (\ref{eq:rcomb}) depends only on the winding numbers of
the zigzag paths.

In Lemma \ref{lem:pm1construct} we will exhibit for any set of winding
numbers a dimer for which the number of independent solutions to
(\ref{eq:rvert}) and (\ref{eq:rface}) is the number of zigzag paths minus
one.
\end{proof}

Lemma \ref{lem:pathsminus1}
tells us how to solve (\ref{eq:rvert}) and (\ref{eq:rface}) for a large class
of dimers, many of which are not properly ordered.
It is interesting to note that the second move shown in 
Figure \ref{fig:rearrange} does not change either $a$ or $\sum_e (1-R(e))$.
The first move also leaves $a$ and $\sum_e (1-R(e))$ invariant
in the case where the two zigzag paths are oppositely oriented
(the charges of the introduced vertices sum to two and
$(R_1-1)^3+(R_2-1)^3=(R_1+R_2-2)\left[ (R_1-1)^2 - (R_1-1)(R_2-1) + (R_2-1)^2 \right]$).
When the zigzag path diagram corresponds to a dimer,
$\sum_e (1-R(e))$ is the number of faces in the dimer.

\begin{proof}[Proof of Theorem \ref{thm:bdysymmetry}]
First we will show that the $\delta_M$ are solutions to (\ref{eq:rface}).
Suppose a face $f$ with $2n$ sides had $n$ of those sides in a boundary perfect
matching $M$.  (A side of a face is an edge of the face along with a normal
pointing into the face.  If a face borders itself then the bordering edge is
part of two different sides of the face.  If a self-border edge is in
a perfect matching, then we count two sides of $f$ in that perfect matching.)
From Corollary \ref{cor:bdymatchings},
we know that we can get from $M$ to any other boundary
perfect matching by flipping zigzag paths.  Note that this operation leaves
invariant the number of sides of each face in the perfect matching.
Therefore every boundary perfect
matching has $n$ sides of $f$.  So every node of $f$ selects one of the two
adjacent sides of
$f$ for all boundary perfect matchings.  By Lemma \ref{lem:edgeinpm}
we know that every edge is in some corner perfect matching.
So the only edges belonging to any node of $f$ are the adjacent sides of $f$.
Therefore, as we move along the boundary of the face we are following
a zigzag path.  But then we have a zigzag path with zero winding, which
violates proper ordering.
So the assumption that a face with $2n$ sides can have $n$ sides in a boundary
perfect matching must be false.  Therefore a face with $2n$ sides can have at
most $n-1$ sides in a boundary perfect matching.  Sum this inequality over all
faces:
\begin{equation}
  \sum_f \sum_{s \in f \cap M} 1 \le \sum_{f} \left [\left( \sum_{s \in f} \frac{1}{2} \right) - 1 \right]
\end{equation}
where $f$ runs over faces and $s$ runs over sidess.
Now reverse the order of the sums:
\begin{eqnarray}
  \sum_{s \in M} \sum_{f \owns s} 1 & \le & \sum_s \sum_{f \owns s} \frac{1}{2} - F \\
  V & \le & (2E)\left(\frac{1}{2}\right) - F \\
    V & \le & E-F.
\end{eqnarray}
Since we know $V=E-F$, equality must
have held in each case.  So (\ref{eq:rface}) is satisfied by boundary
perfect matchings.

The difference between any two boundary perfect matchings is a sum of functions
$\delta_Z$, where $Z$ is a zigzag path and the value $\delta_Z$ alternates
between $2$ and $-2$ on $Z$ and is zero outside of $Z$.  The only relation
obeyed by the $\delta_Z$ is that they sum
to zero.  So the dimension of the space of
solutions to (\ref{eq:rvert}) and (\ref{eq:rface}) that we have found
equals the number of zigzag paths minus one.
By Lemma \ref{lem:pathsminus1}, there can be no more solutions.
\end{proof}

When some of the boundary points of the toric diagram are not corners,
there are many sets of perfect matchings that form a basis for the
solutions to $(\ref{eq:rvert})$ and $(\ref{eq:rface})$.
We will construct a basis by
associating each segment of the boundary of the toric diagram with a zigzag
path, and choosing one perfect matching at each boundary point so that the
difference between two consecutive perfect matchings is the zigzag path
corresponding to the segment between them.  Write
$R = \sum_i \lambda_i \delta_{M_i}$, where $M_i$ are the perfect matchings in
the basis and the $\lambda_i$ are real numbers.

Butti and Zaffaroni \cite{Butti:2005vn}
also noted that in many cases each edge that is a
positively oriented intersection of a zigzag path $Z_r$ with
another zigzag path $Z_s$ occurs in the perfect
matchings in $cc(r,s)$, the counterclockwise segment from $r$ to $s$,
while a negatively oriented intersection of $Z_r$ with $Z_s$
occurs in the perfect matchings not in $cc(r,s)$.  In this
case, the value of $R-1$ for a positively oriented intersection of $Z_r$
with $Z_s$ is $2\left(\sum_{i \in cc(r,s)} \lambda_i \right) -1$.  For a
negatively oriented intersection the value of $R-1$ is
$2\left(\sum_{i \notin cc(r,s)} \lambda_i \right) -1$, which equals
$-\left[2\left(\sum_{i \in cc(r,s)} \lambda_i \right) -1\right]$ since
$\sum_i \lambda_i=1$.  So then the total contribution to $\sum_e (R-1)^3$
from the intersections of $Z_r$ with $Z_s$ is
$\left(\mathbf{w}_r \wedge \mathbf{w}_s \right)
 \left[2\left(\sum_{i \in cc(r,s)} \lambda_i \right) -1 \right]^3$.  Hence
(\ref{eq:thooft}) can be rewritten as
\begin{equation} \label{eq:buttia}
a = \frac{9N^2}{32} \left[ F + \sum_{r<s} \left(\mathbf{w}_r \wedge \mathbf{w}_s
\right)
\left( 2\left(\sum_{i \in cc(r,s)} \lambda_i \right) -1 \right)^3 \right].
\end{equation}
\begin{prop} \label{prop:buttia}
  If a dimer has properly oriented nodes, then
  it is the case that all positively
  (resp. negatively) oriented intersections of $Z_r$ with $Z_s$ are in
  precisely the perfect matchings that are in $cc(r,s)$ (resp. $cc(s,r)$).
  Hence \ref{eq:buttia} holds for properly ordered dimers.
\end{prop}

\begin{proof}
Assume that the dimer has properly ordered nodes.
As we go around
the toric diagram, the perfect matching switches from containing an edge
$e$ to not containing it only if we changed the perfect matching by a
zigzag path containing $e$.  So each intersection of $Z_r$ with $Z_s$
occurs in either the perfect matchings in $cc(r,s)$ or the perfect matchings
in its complement.  From Lemma \ref{lem:edgeinpm}
we know that the positively oriented intersections
are in the corners of $cc(r,s)$ and the negatively oriented intersections 
are not.
\end{proof}

A particularly nice rearrangement of (\ref{eq:buttia}) that we will find
useful is
\cite{Benvenuti:2006xg, Lee:2006ru}
\begin{equation}
\label{eq:area}
a = \frac{9N^2}{4} \sum_{ijk} \area(P_i P_j P_k) \lambda_i \lambda_j
\lambda_k.
\end{equation}
where $P_i$ is the point on the toric diagram corresponding to the
$i$th perfect matching.  This formula tells us that
the triangle anomaly of the three symmetries with respective charges
$\delta_{M_i}$, $\delta_{M_j}$, and $\delta_{M_k}$ is
$\frac{N^2}{2} \area(P_i P_j P_k)$.  AdS-CFT predicts that the
$U(1)$ symmetries of the CFT correspond to gauge symmetries in the AdS theory,
and that the triangle anomalies of the CFT should equal the corresponding
Chern-Simons coefficients in the AdS theory \cite{Witten:1998qj}.
The Chern-Simons coefficients are indeed found to be
$\frac{N^2}{2} \area(P_i P_j P_k)$ \cite{Benvenuti:2006xg}.
So the field theory produced by a properly
ordered dimer will have precisely the cubic anomalies
predicted by the AdS theory.  This is strong evidence that properly ordered
dimers are consistent.

\subsection{Unitarity bound}

Gauge invariant scalar
operators in a four-dimensional CFT must have dimension at least one
\cite{Mack-1977}.
We also have the BPS bound $\Delta \geq \frac{3}{2}|R|$, where $\Delta$ is
the dimension of an operator and $R$ is its $R$-charge.  Equality is achieved
in the case of chiral primary operators  \cite{Dobrev-1985}.
So in order for the theory to be physically valid it is necessary that
the gauge invariant chiral primary operators have $R$-charge at least
$\frac{2}{3}$.

\begin{thm} \label{thm:r23}
  If $a$ can be expressed in the form (\ref{eq:area}),
  then there exists an $N$ such that in the dimer theory
  with $N$ colors, each gauge invariant chiral primary operator has $R$-charge
  at least $\frac{2}{3}$.  In particular properly ordered dimers have this
  property.
\end{thm}

\begin{lem} \label{lem:nonnegative}
At the point where $a$ is locally maximized, the weight of each corner
perfect matching is positive, and the weight of the other boundary perfect
matchings is zero.
\end{lem}
\begin{proof}
This follows immediately from equation (4.2) of \cite{Butti:2005vn}.
\end{proof}

\begin{lem}[A. Kato \cite{Kato:2006vx}] \label{lem:muis3a}
If $a$ is given by (\ref{eq:area}), then
at the point where $a$ is locally maximized,
\begin{equation}
  \frac{\partial a}{\partial \lambda_i}=3a.
\end{equation}
\end{lem}
\begin{proof}
We can use Lagrange multipliers to find the local maximum of $a$.
\begin{equation}
\frac{\partial a}{\partial \lambda_i}=\mu \frac{\partial}{\partial \lambda_i}
\sum_j \lambda_j = \mu
\end{equation}
for some constant $\mu$.
Since $a$ is homogeneous of degree three,
\begin{eqnarray}
3 a & = & \sum_i \lambda_i \frac{\partial a}{\partial \lambda_i} \\
& = & \sum_i \lambda_i \mu \\ \label{eq:muis3a}
& = & \mu.
\end{eqnarray}
\end{proof}

\begin{lem}[A. Kato \cite{Kato:2006vx}] \label{lem:onethird}
At the point where $a$ is locally maximized, each $\lambda_i$ is at most
$\frac{1}{3}$.
\end{lem}
\begin{proof}
By Lemma \ref{lem:muis3a},
$3\lambda_i a=\lambda_i \frac{\partial a}{\partial \lambda_i}$.
Since every term of $a$ is degree zero or one in $\lambda_i$, the right-hand
side is simply the terms of $a$ containing $\lambda_i$.
We can see from (\ref{eq:area}) that the coefficient of each term of $a$ is
nonnegative and from Lemma \ref{lem:nonnegative} that each
$\lambda_i$ is nonnegative when $a$ is maximized.  Hence
the sum of the terms of $a$ containing $\lambda_i$ is at most $a$.
Therefore $3 \lambda_i a \le a$, so $\lambda_i \le \frac{1}{3}$.
\end{proof}

\begin{proof}[Proof of Theorem \ref{thm:r23}]
First consider the mesonic operators, which arise as the trace of a product
of of operators corresponding to the edges around a loop of the quiver.
The number of signed crossings between a
loop and a perfect matching of the dimer is an affine function of the perfect
matching's position in the toric diagram.  If the loop has nonzero winding,
then the function is not constant, and  its zero locus is a line.
This line can intersect the corners of the toric diagram at most twice.
Therefore each loop intersects all but at most two of the corner
perfect matchings.  The sums of the weights of those two perfect matchings is
at most $\frac{2}{3}$, and from Lemma \ref{lem:nonnegative} we know
that the non-corner matchings have weight zero.
The sum of the weights of the perfect matchings that do intersect the loop is
then at least $\frac{1}{3}$.  So the loop has $R$-charge at least $\frac{2}{3}$.
The $R$-charge of a loop with zero winding is twice the number of
intersections it has with any perfect matching.
Every edge is in at least one perfect matching so this number must be positive.
So a loop with zero winding has $R$-charge at least $2$.

The theory also has baryonic operators.  If the gauge groups
are $SU(N)$ then these operators are the $N$th exterior powers the
bifundamental fields.
Each edge of the dimer is contained in at least one
corner perfect matching by Lemma \ref{lem:edgeinpm},
and we know from Lemma \ref{lem:nonnegative} that each
corner of the toric diagram has a positive contribution to the $R$-charge.
So each dimer edge has positive $R$-charge.
For sufficiently large $N$,
the corresponding baryonic operator will have $R$-charge at least $\frac{2}{3}$.
\end{proof}

\section{Bounds on $a$}

\subsection{Bounds on $a$ for toric theories}

We can use (\ref{eq:area}) to establish bounds for $a$.  In this section
we let the indices $ijk$ of the perfect matchings run over the corner perfect
matchings only, since we know from Lemma \ref{lem:nonnegative} that the
non-corner perfect matchings have weight zero.

\begin{thm} \label{thm:bound} Let $N$ be the number of colors of each gauge group,
and let $K$ be the area of the toric diagram
(which is half the number of gauge groups).  Then
\begin{equation} \label{eq:bounds}
\frac{27N^2K}{8 \pi^2} < a \le \frac{N^2K}{2}.
\end{equation}
The upper bound is achieved iff the toric diagram is a triangle, and the
lower bound is approached as the toric diagram approaches an ellipse.
\end{thm}

\begin{proof}[Proof of the lower bound of Theorem \ref{thm:bound}]
The polar body $X^*_R$ of a convex polygon $X$ with respect to the point $R$
is defined as the set of points $Q$ satisfying
$\overrightarrow{RQ} \cdot \overrightarrow{RP} \le 1$ for all $P \in X$.
Recall that maximizing $a$ is equivalent to minimizing the volume of
a slice of the dual toric cone \cite{Martelli:2005tp, Butti:2005vn}.
More specifically, if $\vec{r}$ is the three-dimensional Reeb vector, then
$\frac{9N^2}{8a}$ is the volume of the set of points $\vec{x}$ in the dual cone
satisfying
$\vec{r} \cdot \vec{x} \le 3$.  The cross section of the dual cone in
the plane $\vec{r} \cdot \vec{x} = 3$ is the polar body of the toric diagram
with respect to the Reeb vector (considered as a point in the plane of the
toric diagram).
If we call the toric diagram $X$, then
$\frac{27N^2}{8a} = \inf_{R \in X} \area(X^*_R)$.
Then the statement of the lower bound is equivalent to
$\area(X) \inf_{R \in X} \area(X^*_R) < \pi^2$.
The result $\area(X) \inf_{R \in X} \area(X^*_R) \le \pi^2$
was proved by Blaschke \cite{Blaschke-1917,Petty-1985}; equality occurs in
the case of an ellipse.  Since the toric diagram is a polygon, it cannot
be perfectly elliptical and hence equality does not hold.
\end{proof}

We will need to use the following results for the proof of the upper bound.

\begin{prop}[A. Kato \cite{Kato:2006vx}] \label{prop:localismax}
The local maximum of $a$ is the overall maximum of $a$ in the region
$\lambda_i \ge 0, \sum_i \lambda_i = 1$.
\end{prop}

\begin{prop}[A. Butti and A. Zaffaroni \cite{Butti:2005vn}]
\label{prop:maxpoint}
Let $R$ be a point in the interior of the toric diagram.  Define
\begin{eqnarray} \label{eq:reebf}
f_i & = & \frac{\area(P_{i-1} P_i P_{i+1})}{\area(P_{i-1} P_i R) \area(P_i P_{i+1} R)} \\ \label{eq:reebs}
S & = & \sum_i f_i \\ \label{eq:reeblambda}
\lambda_i & = & f_i/S.
\end{eqnarray}
Then the following results hold:
\begin{eqnarray}
R & = & \sum_i \lambda_i P_i \label{eq:reebcm} \\
a & = & \frac{27N^2}{2S} \label{eq:reebrecip}.
\end{eqnarray}
Furthermore, when $R$ is the Reeb vector and
the $\lambda_i$ are given by (\ref{eq:reeblambda}), $a$ is locally maximized
(over all choices of $\lambda_i$, not just those of the form
(\ref{eq:reeblambda})).
\end{prop}

\begin{proof}[Proof of the upper bound of Theorem \ref{thm:bound}]
We use induction on the number of
corners of the toric diagram.  If the toric diagram is a triangle, then $a$
is maximized when each $\lambda_i$ is $\frac{1}{3}$.  So
$a=\frac{9N^2}{4}K(3!)\left(\frac{1}{3}\right)^3=\frac{N^2 K}{2}$.

Assume the toric diagram has more than three corners.
Let $\lambda_i^M$
be the values of $\lambda_i$ for which $a$ is locally maximized.
Choose a particular $i$ and let $\lambda_i^D=0$,
$\lambda_{i+1}^D=\lambda_i^M+\lambda_{i+1}^M$, and $\lambda_j^D=\lambda_j^M$
for all other $j$.  We will define $a^M=a|_{\lambda^M}$, $a^D=a|_{\lambda^D}$,
and $\Delta a = a^D - a^M$.
Since $a$ has degree one in each individual $\lambda_j$,
\begin{equation} \label{eq:dalinear}
\Delta a = 
\left.\frac{\partial a}{\partial \lambda_i}\right|_{\lambda^M} (-\lambda_i^M) +
\left.\frac{\partial a}{\partial \lambda_{i+1}}\right|_{\lambda^M} \lambda_i^M + 
\left.\frac{\partial^2 a}{\partial \lambda_i \partial \lambda_{i+1}}\right|_{\lambda^M} (-\lambda_i^M)
(\lambda_i^M)
\end{equation}
Recall that since $a$ is initially maximized,
$\frac{\partial a}{\partial \lambda_i}|_{\lambda^M}=
\frac{\partial a}{\partial \lambda_{i+1}}|_{\lambda^M}$ 
and hence the first two terms of (\ref{eq:dalinear}) cancel.
Now use (\ref{eq:area}) to expand the last term:
\begin{equation}
\Delta a = -\frac{27N^2}{2} (\lambda_i^M)^2 \sum_j \lambda_j^M \area(P_i P_{i+1} P_j).
\end{equation}
Since all of the $P_j$ are on the same side of the line $P_i P_{i+1}$,
\begin{equation}
\Delta a = -\frac{27N^2}{2}(\lambda_i^M)^2 \area(P_i P_{i+1} R)
\end{equation}
where $R$ is the weighted center of mass of the $P_j$ with weights
$\lambda_j^M$.
Now apply Proposition \ref{prop:maxpoint}.  We can write
\begin{eqnarray}
\Delta a & = & -\lambda_i^M \frac{27N^2\area(P_{i-1} P_i P_{i+1})}{2S \area(P_{i-1} P_i R)} \\
& = & -\lambda_i^M a^M \frac{\area(P_{i-1} P_i P_{i+1})}{\area(P_{i-1} P_i R)}.
\end{eqnarray}
Since $\sum_i \lambda_i^M = 1$ and $\sum_i \area(P_{i-1} P_i R)=K$, there
must be some $i$ for which
$\frac{\lambda_i^M}{\area(P_{i-1} P_i R)} \le \frac{1}{K}$.
For such an $i$,
\begin{equation}
-\frac{\Delta a}{a^M} \le \frac{\area(P_{i-1} P_i P_{i+1})}{K}
\end{equation}
Note that $\area(P_{i-1} P_i P_{i+1})$ is the amount by which $K$ would
decrease if we removed $P_i$ from the toric diagram.  Since
$\lambda_i^D=0$, the $\lambda_j^D$ are a valid choice of weights
for the toric diagram with $P_i$ removed.
Then
\begin{equation}
-\frac{\Delta a}{a^M} \le -\frac{\Delta K}{K}.
\end{equation}
Therefore $\frac{a^D}{K+\Delta K} \ge \frac{a^M}{K}$.
By Proposition \ref{prop:localismax}
the local maximum value of $a$ for the new toric diagram
is at least as large as $a^D$.
We want to show that it is strictly larger, or equivalently,
that $\lambda^D_j$ do not locally maximize $a$ for the new toric diagram.
Recall from Lemma \ref{lem:muis3a} that $a$ is locally maximized when
$\frac{\partial a}{\partial \lambda_{i+1}}=3a$.  Hence $a$ will continue
to be maximized only if $\Delta \frac{\partial a}{\partial \lambda_{i+1}}=3 \Delta a$.   Once again we use the fact that $a$ is degree one in each 
individual $\lambda_j$:
\begin{eqnarray}
\Delta \frac{\partial a}{\partial \lambda_{i+1}} & = &
\left.\frac{\partial^2 a}{\partial \lambda_i \partial \lambda_{i+1}}\right|_{\lambda^M} (-\lambda_i^M) \\
& = & \frac{\Delta a}{\lambda_i^M}.
\end{eqnarray}
Hence $a$ can continue to be maximized only if
$\lambda_i^M=\frac{1}{3}$.  But $\lambda_{i+1}^M$ is positive
(since we chose to let our indices enumerate corner perfect matchings only), so
$\lambda_{i+1}^D=\lambda_i^M + \lambda_{i+1}^M>\frac{1}{3}$.  By
Lemma \ref{lem:onethird}, $\lambda_{j}^D$ cannot be the local maximum point.
By the induction hypothesis, the new $\frac{a}{K}$ is at most $\frac{1}{2}$,
so the old $\frac{a}{K}$ must be smaller than $\frac{1}{2}$.
\end{proof}

\subsection{Comparison to non-toric field theories}

Let us consider how we might formulate a similar bound for non-toric CFTs.
We need to decide how to interpret $K$ in the non-toric case.  If seems natural
to replace $2N^2 K$ with the sum of the squares of the numbers of colors of each
gauge group.

\begin{TABLE} { \label{tab:karpov}
\begin{tabular}{|c|c|c|c|c|} \hline
Equation in \cite{Karpov-1997} & (x,y,z) & $(9-n)(\alpha x^2 + \beta y^2 + \gamma z^2)$ \\ \hline
(1) & $(1,1,1)$ & 27 \\
(1) & $(1,1,2)$ & 54 \\
(1) & $(1,2,5)$ & 270 \\
(2) & $(1,1,1)$ & 32 \\
(3) & $(1,1,1)$ & 36 \\
(4) & $(1,2,1)$ & 50 \\
(5) & $(1,1,1)$ & 32 \\
(6.1) & $(1,1,1)$ & 27 \\
(6.2) & $(2,1,1)$ & 36 \\
(7.1) & $(2,2,1)$ & 32 \\
(7.2) & $(2,1,1)$ & 32 \\
(7.3) & $(3,1,1)$ & 36 \\
(8.1) & $(3,3,1)$ & 27 \\
(8.2) & $(4,2,1)$ & 32 \\
(8.3) & $(3,2,1)$ & 36 \\
(8.4) & $(5,2,1)$ & 50 \\
\hline
\end{tabular}
\caption{The values of
$\frac{54N^2K}{a}=(9-n)(\alpha x^2 + \beta y^2 + \gamma z^2)$ for some of
the quiver gauge theories defined in \cite{Karpov-1997}.  Note that
the equations in \cite{Karpov-1997} have infinitely many solutions
(which can be seen by observing that if we fix one of $x,y,z$ we get a form
of Pell's equation), so there exist theories with arbitrarily large
$\frac{54N^2K}{a}$.}
} \end{TABLE}

Let's look at the values of $\frac{a}{N^2K}$ for a cone over a del Pezzo
surface.  Reference \cite{Karpov-1997} lists some quiver gauge theories
that are dual to these Calabi-Yaus.  In their notation, the sum of the squares
of the number of colors is $\alpha x^2 + \beta y^2 + \gamma z^2$.
We can compute $a$ by looking at the $AdS$ dual theory.
References \cite{Gubser:1998vd, Henningson:1998gx} tell us that
$\frac{\pi^3}{4a}$ is the volume of the horizon, and
\cite{Bergman:2001qi} tells us that the volume of the real cone over $dP_n$ is 
$\frac{\pi^3(9-n)}{27}$.  So $a=\frac{27}{4(9-n)}$.
So then $\frac{54N^2K}{a}=(9-n)(\alpha x^2 + \beta y^2 + \gamma z^2)$, and
the bound ($\ref{eq:bounds}$) for toric theories is then equivalent to
$27 \le (9-n)(\alpha x^2 + \beta y^2 + \gamma z^2) < 4 \pi^2$.
From Table \ref{tab:karpov} we see that the toric upper bound on
$\frac{54 N^2 K}{a}$ is not true for
all quiver gauge theories.  In fact, $\frac{N^2 K}{a}$ can be arbitrarily large.
Equation (1) of \cite{Karpov-1997} is $x^2+y^2+z^2=3xyz$.  If we set $z=1$ then
we have a Pell's equation in $x$ and $y$ and there are infinitely many solutions.
On the other hand, $27 \le (9-n)(\alpha x^2 + \beta y^2 + \gamma z^2)$
still holds for all of the theories considered in \cite{Karpov-1997}.  It would
be interesting to know if the inequality holds more generally.

\section{\label{sec:deleting}Merging zigzag paths}

\subsection{Deleting an edge of the dimer}

Theorem \ref{thm:zzp} says that, if a dimer is properly ordered, then
we can determine its toric diagram from the windings of its zigzag paths.
As we mentioned in section \ref{sec:zzp}, Hanany and Vegh \cite{Hanany:2005ss}
and Stienstra \cite{stienstra-2007} have previously made proposals for drawing
a dimer with given zigzag winding numbers, but their procedures are
impractical for large dimers because of the large amount of trial and error
required.

Partial resolution
\cite{Douglas:1997de, Morrison:1998cs, Feng:2000mi, GarciaEtxebarria:2006aq}
has previously been suggested as a method of
determining the dimer from the quiver
\cite{Feng:2000mi, Feng:2002zw, Hanany:2005ve}.
It involves starting with a toric diagram whose dimer model is known and
introducing Fayet-Iliopoulos terms that Higgs some of the fields and remove
part of the toric diagram to create a new diagram.  However, as is the case
with the Fast Inverse Algorithm, the previous proposals involving partial
resolution suffered from being computationally infeasible.

In this section, we will explore how certain operations on the dimer
affect its zigzag paths.  These operations can be interpreted as
partial resolutions.  We will later use these operations to construct
an algorithm for drawing a properly ordered dimer with given winding numbers
that requires no trial and error.

\begin{FIGURE}
{ 
  \label{fig:merging}
  \centerline{
    \epsfig{figure=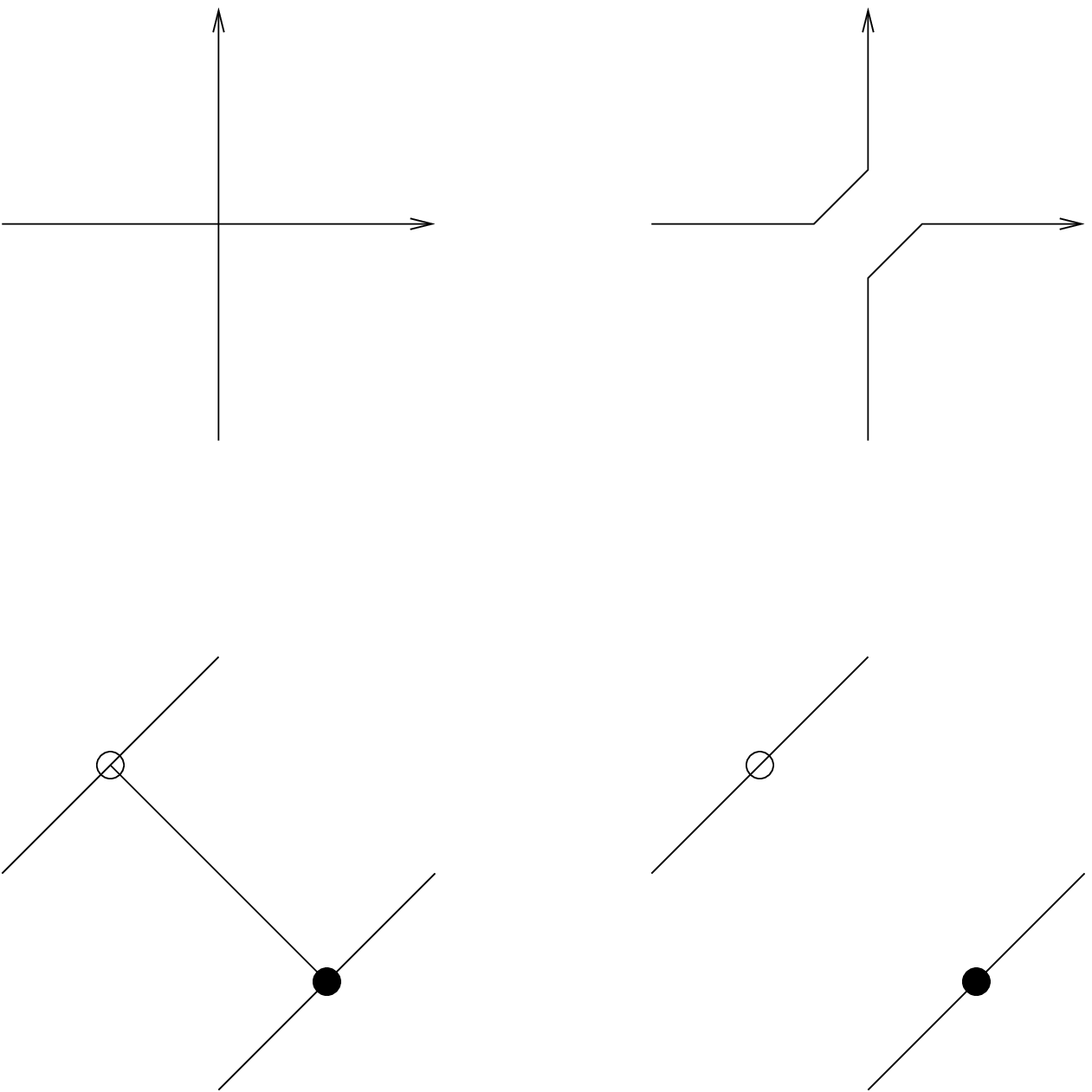, width=3.125in}
  }
\caption{Top: Merging two zigzag paths by deleting the
intersection between them.  Bottom: The effect on the dimer.}
}
\end{FIGURE}

One operation that we can perform is to remove an intersection
of two zigzag paths (or equivalently, delete an edge of the dimer).
The operation has the effect of merging the two paths into a single path.
An example is shown in figure \ref{fig:merging}.
In physical terms, we are Higgsing away the edge by turning on
Fayet-Iliopoulos parameters for the adjacent faces.
This is an example of partial resolution of the toric
singularity \cite{Douglas:1997de, Morrison:1998cs, Feng:2000mi, GarciaEtxebarria:2006aq}.
We will always merge paths that intersect just once.  In the following
we will sometimes assume that the windings of the paths are $(1,0)$ and $(0,1)$;
any other case is $SL_2(\mathbb{Z})$ equivalent to this one.

\subsection{Making multiple deletions}

\begin{FIGURE} {\label{fig:maketwo}
\centerline{
  \epsfig{figure=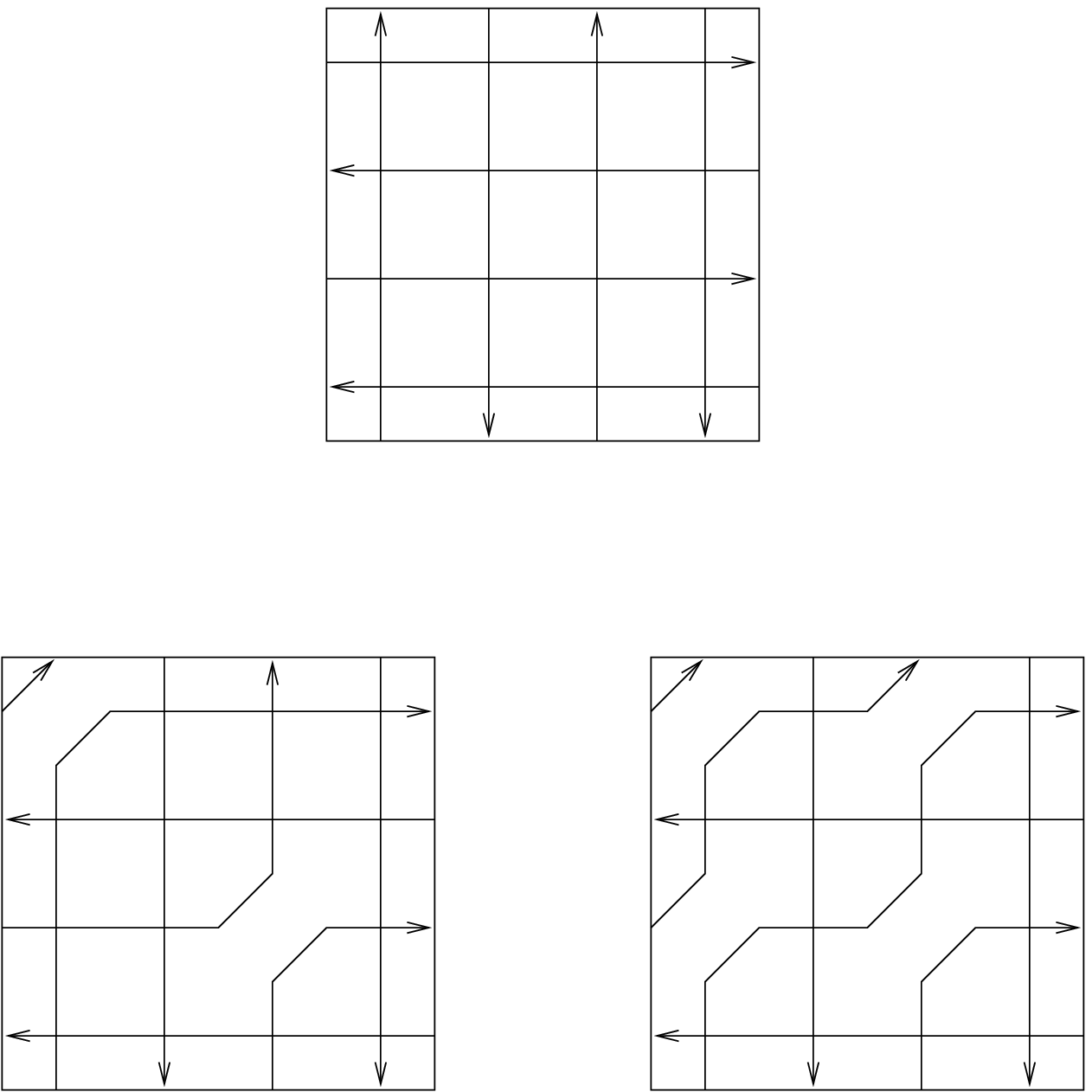,height=3.5in}
}
\caption{Left: An incorrect way of making two $(1,1)$ paths.
They intersect each other, which implies that the adjacent nodes are not
properly ordered.  Right: The correct way of making two $(1,1)$ paths.}
} \end{FIGURE}

Suppose we want to make $n>1$ $(1,1)$ edges from $(1,0)$ and $(0,1)$ edges.
If we make them one at a time, then we would violate the proper ordering of
nodes because we would have $(1,1)$ paths intersecting each other.
We should instead delete all $n^2$ edges between the $n$ $(1,0)$ edges and the
$n$ $(0,1)$ edges.
We will refer to this procedure as Operation I.

\subsection{\label{sec:extra}Extra crossings}

\begin{FIGURE}
{
  \label{fig:extracrossing}
  \centerline{
    \epsfig{figure=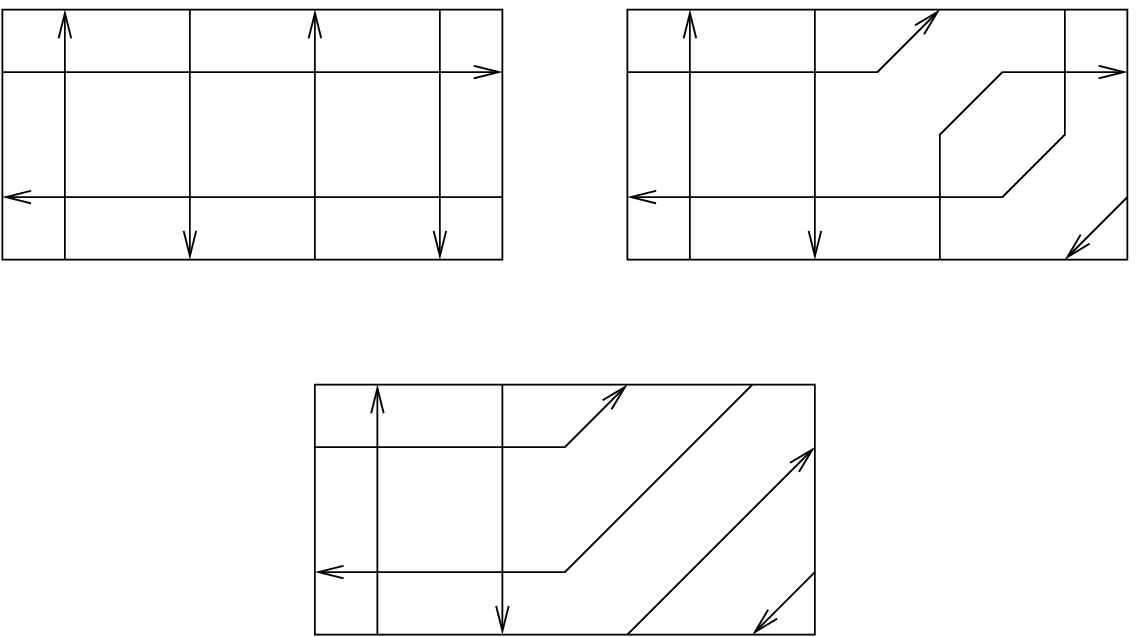,width=4.5in}
  }
  \caption{The top left diagram has no extra
    crossings.  The top right diagram shows what happens when some zigzag
    paths
    are merged.  The two diagonal paths now have extra crossings with each
    other.
    The bottom diagram shows what happens when we move the two zigzag paths
      past each other; they no longer intersect.}
}
\end{FIGURE}

We mentioned in section \ref{sec:defs} that the number of oriented crossings
between a pair of paths is a function only on their windings.
The number of unoriented crossings is greater than or equal to the absolute
value of the number of oriented crossings.
If equality does not hold then we say that the pair of paths has 
``extra crossings'' .  We say that a diagram has extra crossings if
any pair of its paths does.  There is nothing inherently wrong with extra
crossings, but we may find it desirable to produce diagrams without them.

The edge deletion procedure mentioned in the previous section
sometimes introduces extra crossings.  An example of this is shown in
Figure \ref{fig:extracrossing}.
We combine a $(1,0)$ zigzag path and a $(0,1)$ to make a $(1,1)$ zigzag path,
and we also combine $(-1,0)$ and $(0,-1)$ paths to make a $(-1,-1)$ path.
The $(1,1)$ path and $(-1,-1)$ path have a positively oriented intersection
coming from the $(0,1)-(-1,0)$ intersection and a negatively oriented
intersection coming from the $(1,0)-(0,-1)$ intersection.  Note that we can
get rid of these crossings by moving the two paths past each other.
In terms of the dimer, moving the paths past each other merges the two
vertices adjacent to a valence two node.
Physically, we are integrating out a mass term.

\begin{FIGURE} {
\label{fig:orientations}
\centerline{
\epsfig{figure=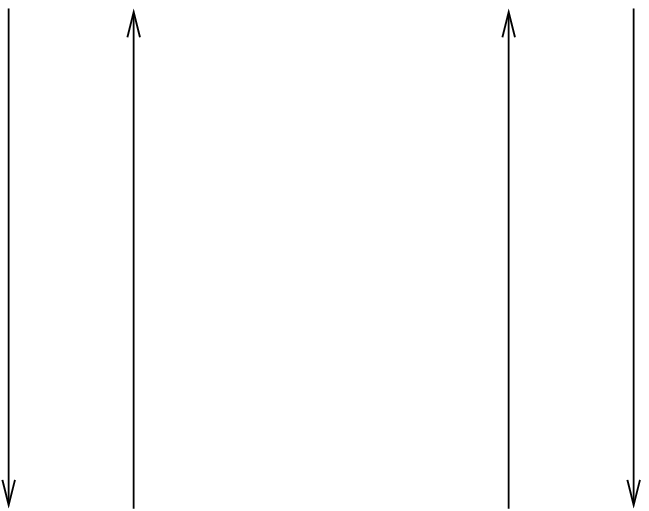,height=1.25in}
}
\caption{Pairs of paths that are positively and
negatively oriented, respectively.}
}
\end{FIGURE}

\begin{FIGURE}
{
\label{fig:movepast}
\centerline{
  \includegraphics[height=3.125in]{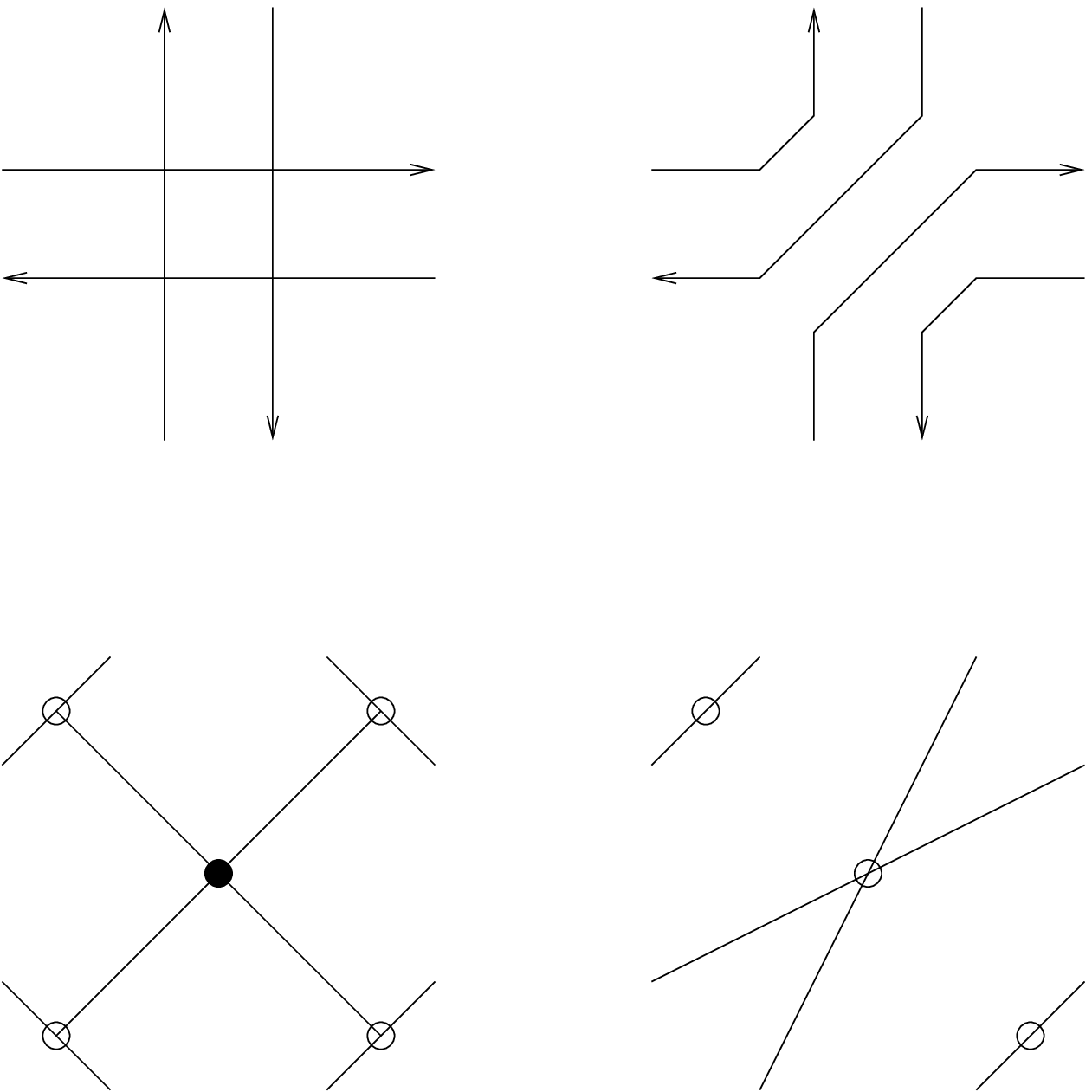}
}
\caption{Making $(1,1)$ and $(-1,1)$ paths from horizontal
and vertical paths in the zigzag path picture and the dimer picture.}
}
\end{FIGURE}

We define a pair of zigzag paths to be an ``opposite pair'' if they have
opposite winding numbers, they do not intersect, and they bound a region
containing no crossings.
Also, we define the orientation of an opposite pair
to be positive if the area containing no crossings is to the left of an observer
traveling along one of the paths, and negative if the area is on the right.
(See figure \ref{fig:orientations}.)  We have just seen how to take
a pair negatively oriented horizontal paths and a pair of negatively oriented
vertical paths and turn them into a pair of negatively oriented diagonal paths.
Similarly we can turn a pair of positively oriented horizontal paths and a pair
of positively oriented vertical paths into a pair of positively oriented
diagonal paths.  In terms of dimers, this operation takes a node of valence
four, deletes two opposite edges, and merges the other endpoints of the two
remaining edges.  Figure \ref{fig:movepast} shows the operation in terms of
both zigzag paths and dimers.

\begin{FIGURE} {\label{fig:multipledeletions}
\centerline{\includegraphics[width=4.5in]{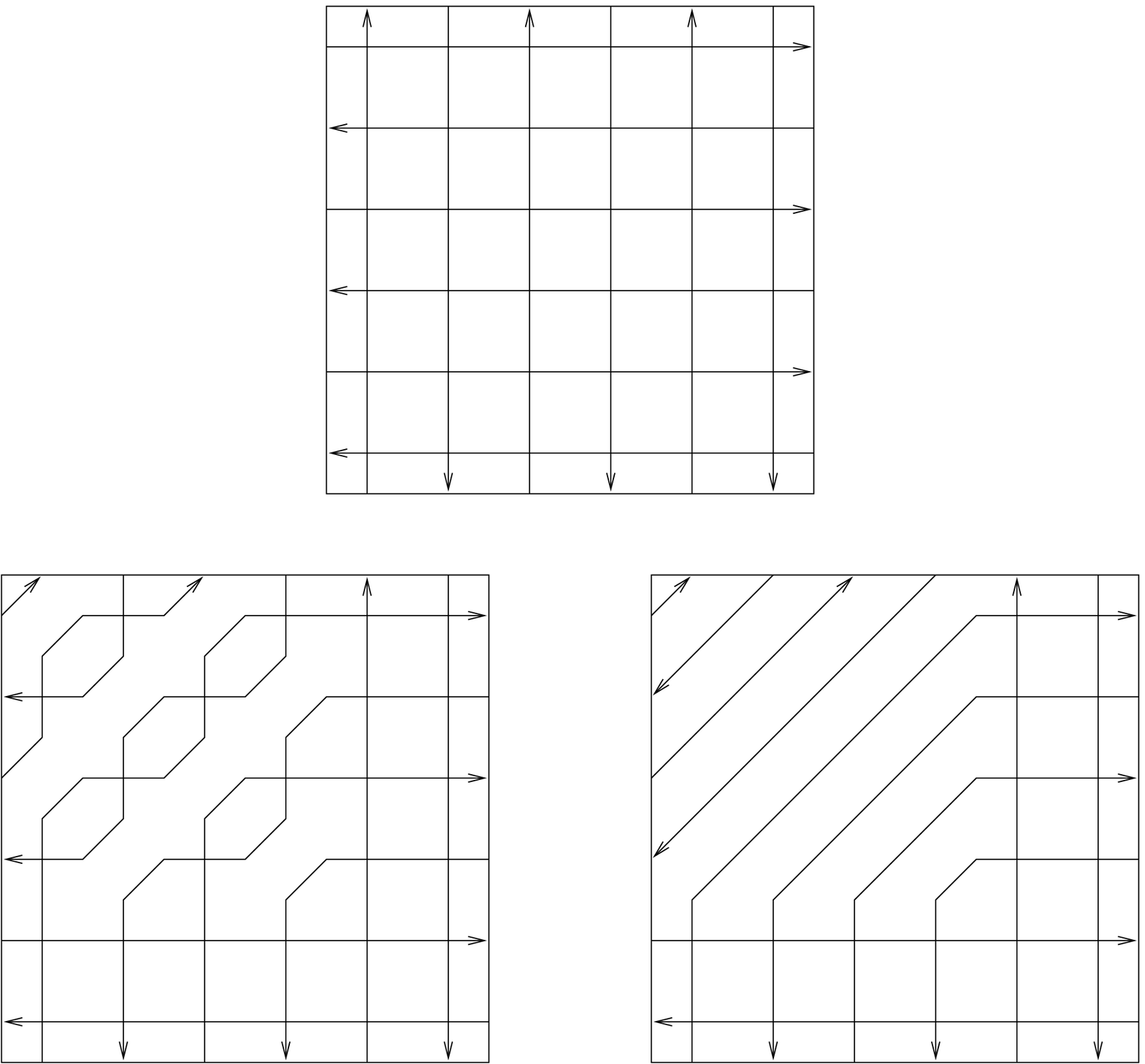}}
\caption{Creating two $(1,1)$ and two $(-1,-1)$
paths from horizontal and vertical paths.  We first merge horizontal
and vertical paths to create diagonal paths, then move the diagonal paths past
each other.}
} \end{FIGURE}

More generally, we can make $n$ $(1,1)$ paths and $n$ $(-1,-1)$ paths
and get rid of their crossings.  An example is given in figure
\ref{fig:multipledeletions}.
We have to untangle each $(1,1)$ path from each $(-1,-1)$ path.  Note that
all $2n$ paths must have the same orientation.  We will call this procedure
Operation II.

\begin{FIGURE} {\label{fig:uneven}
\centerline{\includegraphics[width=4.5in]{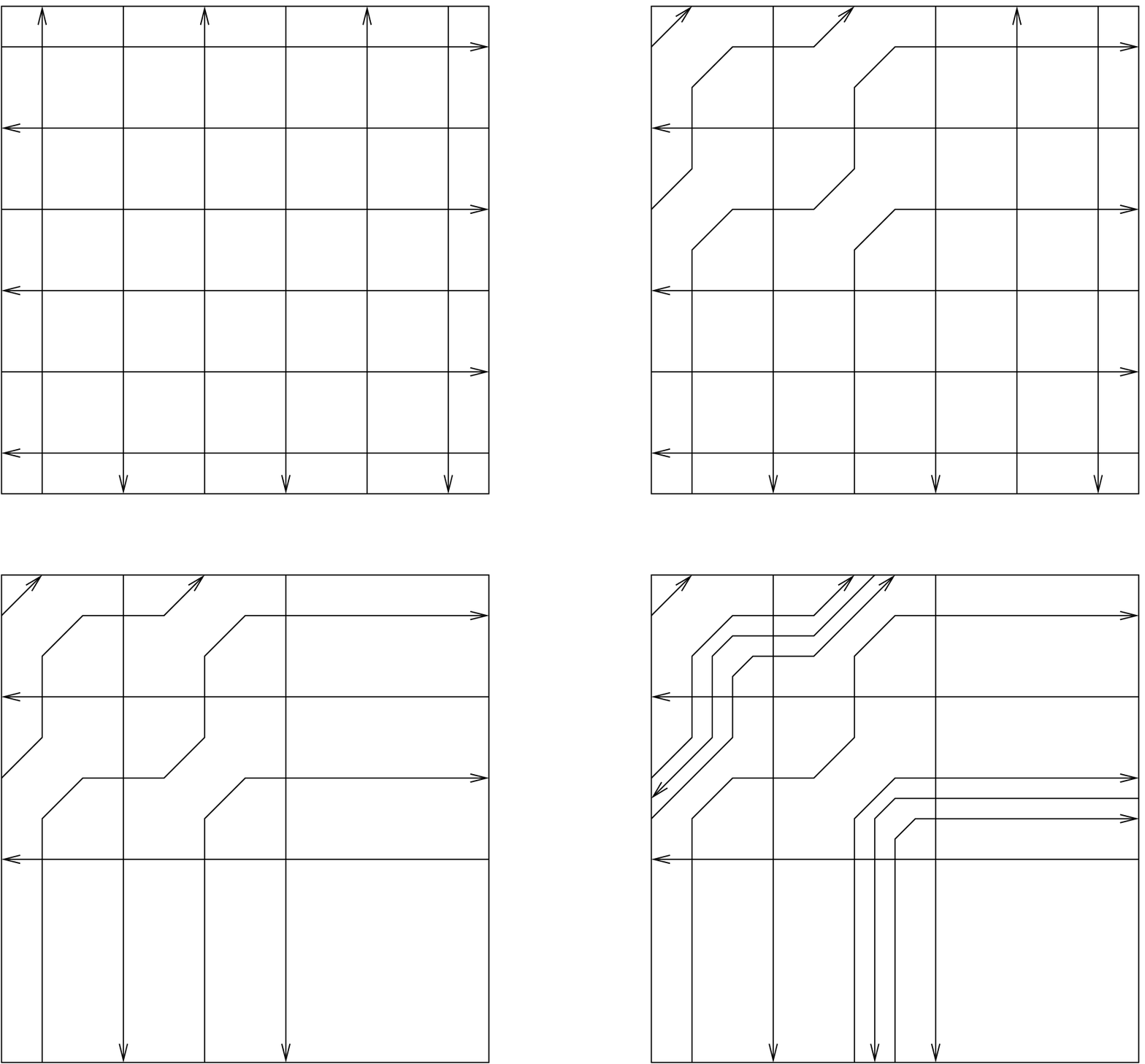}}
\caption{Creating three $(1,1)$ paths and one $(-1,1)$ path
without introducing extra crossings.  We first make two $(1,1)$ paths and
then add a $(1,1)-(-1,-1)$ pair that follows one of those two paths.}
} \end{FIGURE}

\begin{FIGURE} {\label{fig:unevendimer}
\centerline{\includegraphics[width=3.5in]{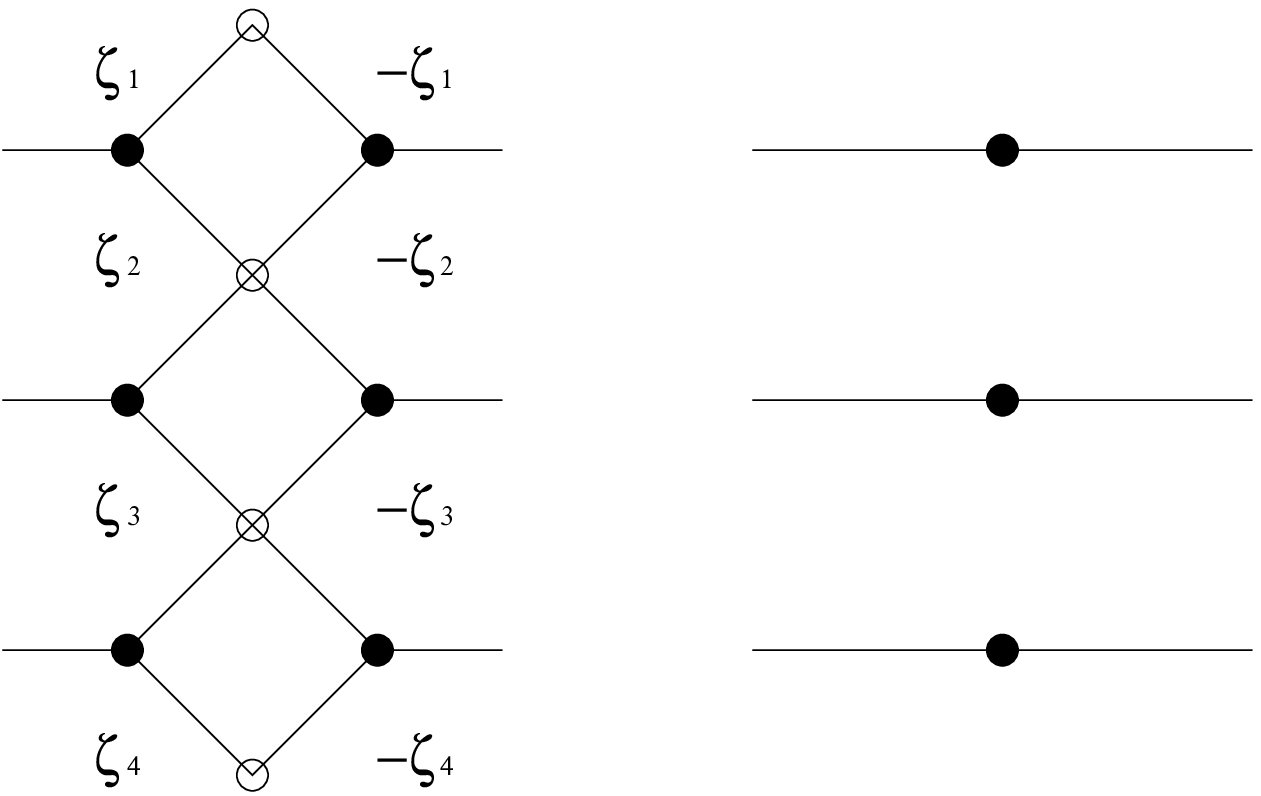}}
\caption{Left: a dimer with a pair of adjacent opposite
zigzag paths.  Right: the dimer with the paths removed.  In physical
terms, we are introducing performing a partial resolution by introducing
Fayet-Iliopoulos parameters for the faces on either
side of the diamonds \cite{Douglas:1997de, Morrison:1998cs, Feng:2000mi, GarciaEtxebarria:2006aq}.  (In particular note that the resolution of the double
conifold in \cite{GarciaEtxebarria:2006aq} is an example of this operation.)
For each pair of faces that meet at one of the points 
in the middle, their FI parameters should sum to zero.  All parameters on the
left should have the same sign.}
} \end{FIGURE}

If we want to create differing numbers of $(1,1)$ and $(-1,-1)$ paths, then we
run into the problem that we cannot pair them all.  We will need to
do something more complicated.  Let $m$ be the number of $(1,1)$ paths we
want to make, and let $n$ be the number of $(-1,-1)$ paths we want to make.
Assume $m>n$.  We first make $m-n$ $(1,1)$ paths.  Now we completely remove
$n$ pairs of adjacent $(1,0)$ and $(-1,0)$ paths and $n$ pairs of
adjacent $(0,1)$ and $(0,-1)$ paths.  Because the pairs are adjacent, the
condition that intersection orientations alternate along a path is preserved.
Now we want to insert $n$
pairs of adjacent $(1,1)$ and $(-1,1)$ paths, and we want to make sure that
there are no extra crossings.  This can be accomplished by making them
follow one of the $m-n$ already existing $(1,1)$ paths.  An example is given
in figure \ref{fig:uneven}.  Figure \ref{fig:unevendimer} shows what removing
or adding a pair of zigzag paths does to the dimer.
This procedure will be called Operation III.

\section{\label{sec:algorithm}An efficient inverse algorithm}

\subsection{Description of the algorithm}

In describing the algorithm we find it useful to draw toric diagrams rotated
90 degrees counterclockwise from their usual presentation.
Our convention will make
the algorithm easier to visualize, because it makes the windings of
the zigzag paths equal to, rather than perpendicular to, the vectors of
the toric diagram edges.

Let $X$ be a toric diagram for which we would like to construct a dimer.
Let $Y$ be
the smallest rectangle with horizontal and vertical sides that contains $X$.
Since $Y$ represents an orbifold of the conifold, we know a dimer for $Y$.
We will modify this dimer until we get a dimer for $X$.

\begin{FIGURE} {\label{fig:tangents}
\centerline{\includegraphics[height=1.25in]{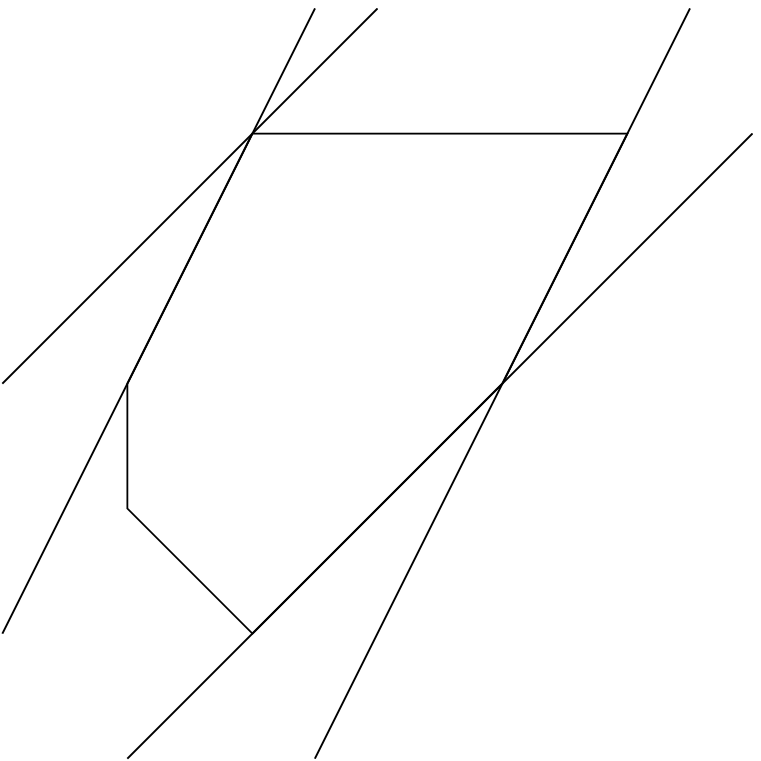}}
\caption{Some tangent lines to a convex polygon.}
} \end{FIGURE}

Before we begin, we need to make the following definition.
A tangent line to a convex polygon $P$ is a line $\ell$ such that
$\ell \cap P \subseteq \partial P$ and $\ell \cap P \ne \emptyset$.
Note that a convex polygon has exactly two tangent lines with a given slope.

\begin{FIGURE} {\label{fig:farey}
\centerline{\includegraphics[width=4.5in]{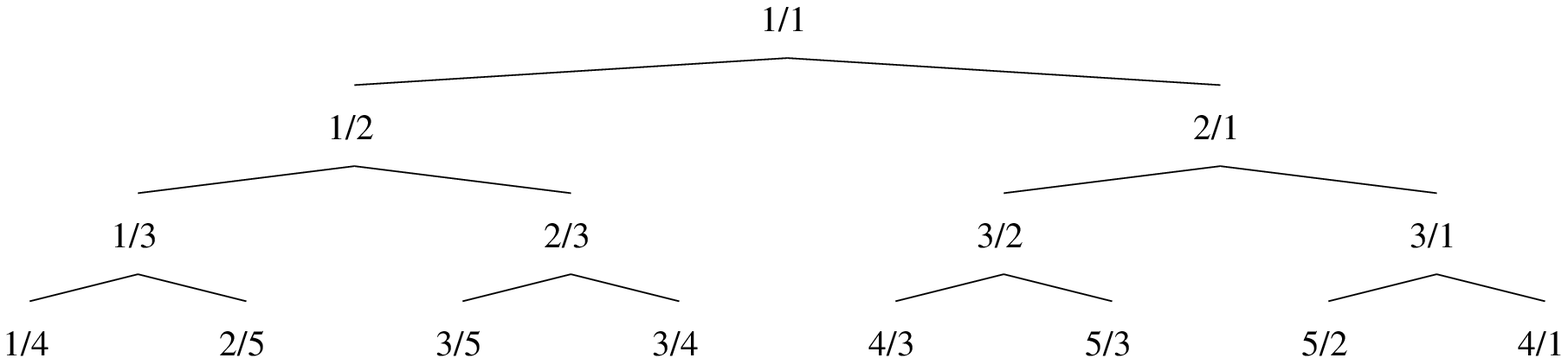}}
\caption{The Farey tree tells us the order in which to
make zigzag paths.  For example, in order to make $(3,4)$ zigzag paths
we first make $(1,1)$ zigzag paths, then $(1,2)$ paths, then $(2,3)$ paths. }
} \end{FIGURE}

We begin by finding the slope one tangent lines to $X$ and cutting $Y$ along
these lines to produce some $(1,1)$ and $(-1,-1)$ paths.
We use Operation I if the number of $(1,1)$ or $(-1,-1)$ paths desired is zero,
Operation II if the numbers are equal, and Operation III if the numbers
are both nonzero and unequal.
Next we want to
cut along the slope $1/2$ tangent lines to $X$ to produce $(2,1)$ and
$(-2,-1)$ paths.  In fact we already know how to do this, because
$SL_2(\mathbb{Z})$ equivalence reduces the problem of making
$(2,1)$ and $(-2,-1)$ paths from $(1,0)$, $(-1,0)$, $(1,1)$, and $(-1,-1)$
paths to the problem of making $(1,1)$ and $(-1,1)$ paths from
$(1,0)$, $(-1,0)$, $(0,1)$ and $(0,-1)$ paths.
Hence we can now cut $Y$ along the slope $1/2$ tangent lines to $X$.
Similarly, we can cut $Y$ along the slope $2$ tangent lines to $X$.
After this, we can make $(3,1)$ paths by combining $(1,0)$ and $(2,1)$ paths,
$(3,2)$ paths by combining $(1,1)$ and $(2,1)$ paths, etc.
We can eventually make paths of all slopes, with the order in which we make
the paths determined by the Farey tree.  (See figure \ref{fig:farey}.)
We can then enumerate over all negative slopes,
starting with $-1$.  When we are finished, we will have a dimer for $X$.

\begin{FIGURE}
{
  \centerline{
    \epsfig{figure=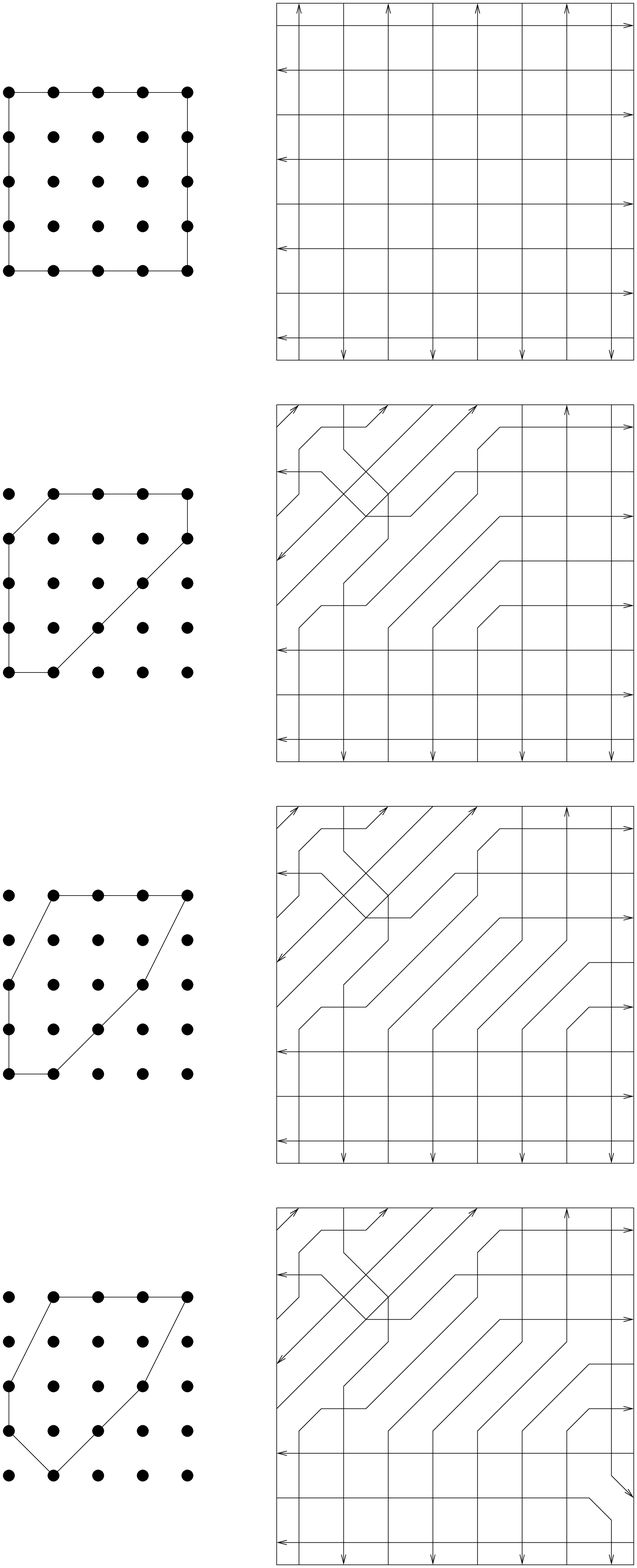, height=6.5in}
  }
  \caption{An example of the algorithm.
  Note that the cut
  made in the second diagram is the same as that of figure
  \ref{fig:uneven},
  although we have drawn it a little differently to make the spacings more
  equal.}
  \label{fig:example}
}
\end{FIGURE}
\begin{FIGURE} {
\centerline{\includegraphics[height=1.25in]{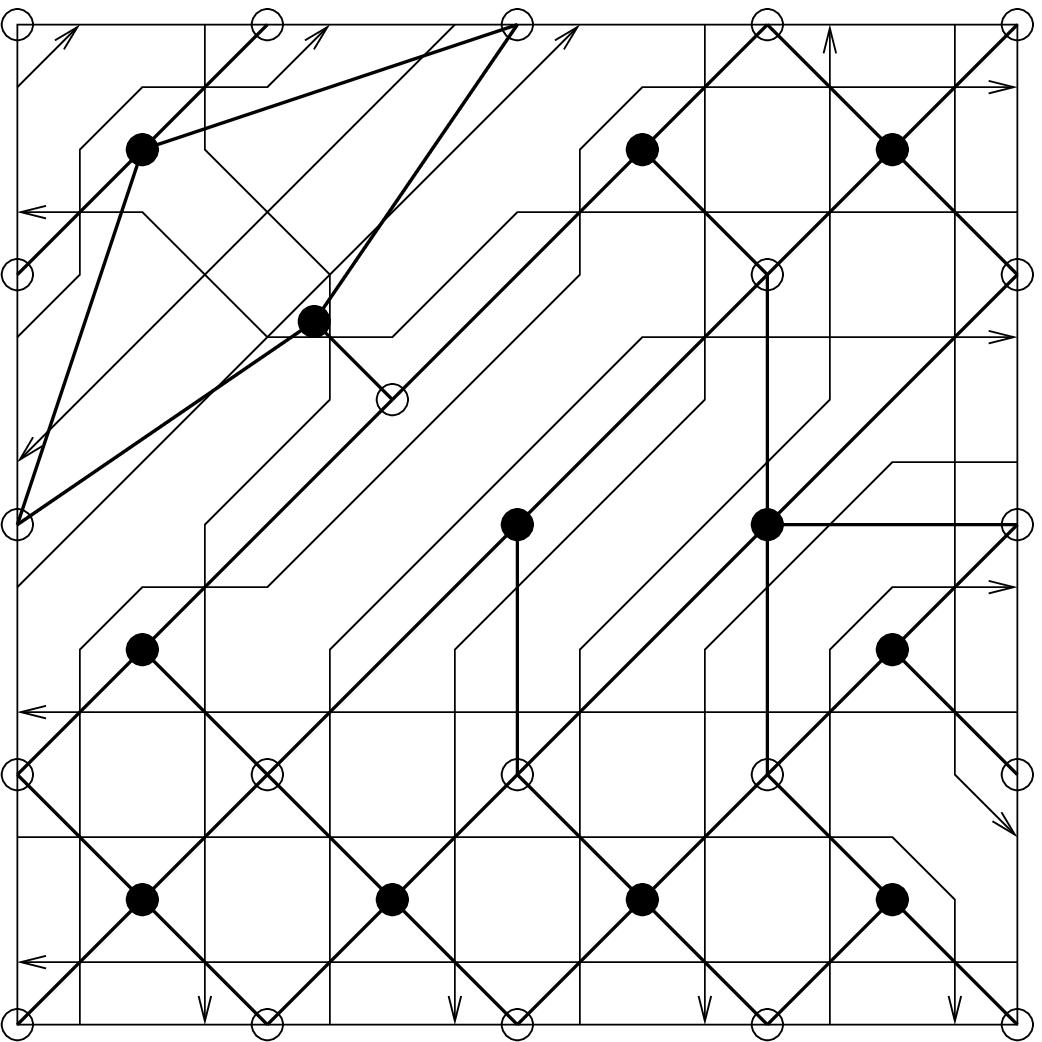}}
\caption{\label{fig:exampledimer} The dimer corresponding to the final
zigzag path diagram in figure \ref{fig:example}.}
} \end{FIGURE}

\begin{FIGURE}
{
  \centerline{
    \epsfig{figure=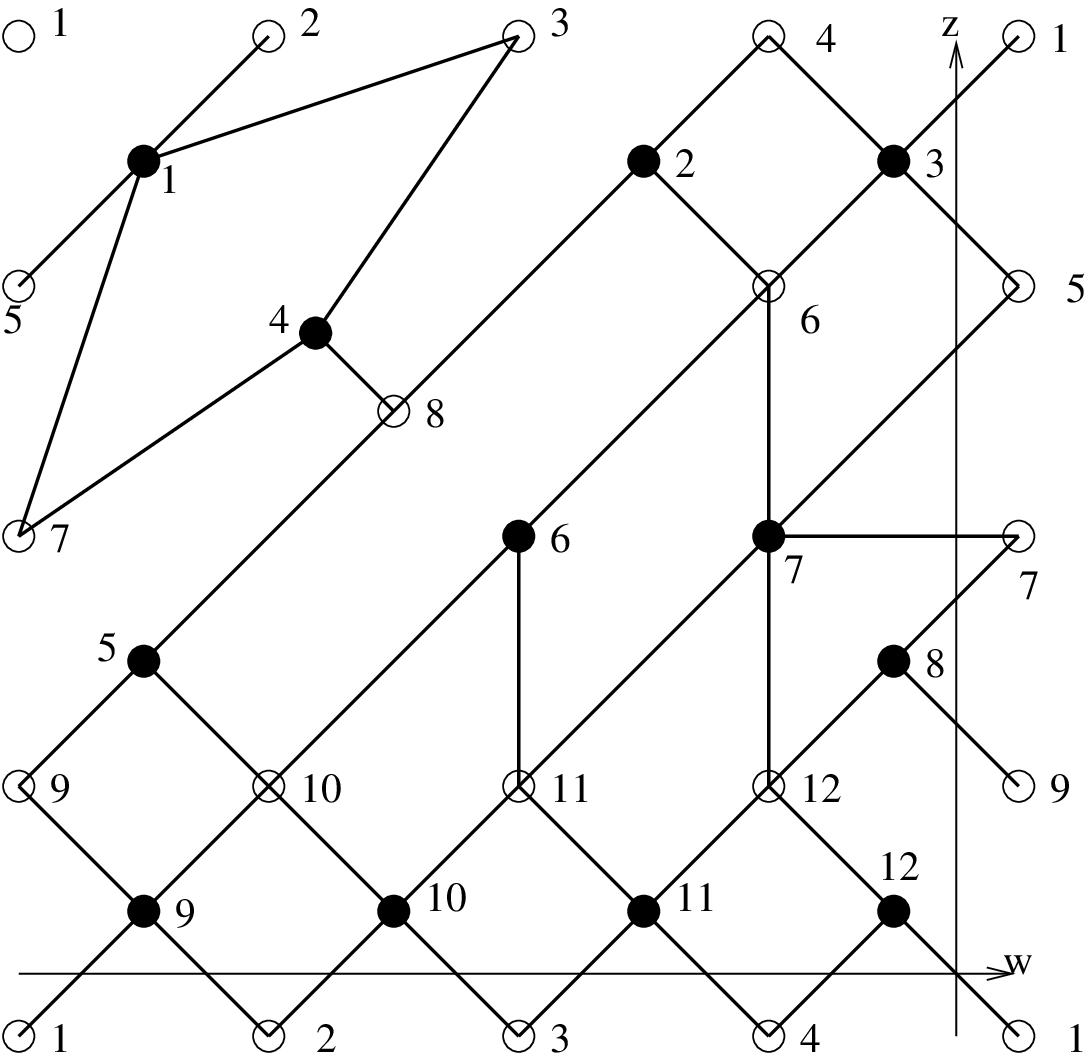, height=1.25in}
  }
  \centerline{}
  \centerline{
    $
    \left(
    \begin{array}{cccccccccccc}
      0 & 0 & z & 0 & 0 & 0 & 0 & 0 & -w & 0 & 0 & wz \\
      1 & 0 & 0 & 0 & 0 & 0 & 0 & 0 & w & -w & 0 & 0 \\
      1 & 0 & 0 & 1 & 0 & 0 & 0 & 0 & 0 & w & -w & 0 \\
      0 & 1 & 1 & 0 & 0 & 0 & 0 & 0 & 0 & 0 & w & -w \\
      -1 & 0 & z & 0 & 0 & 0 & z & 0 & 0 & 0 & 0 & 0 \\
      0 & 1 & -1 & 0 & 0 & 1 & 1 & 0 & 0 & 0 & 0 & 0 \\
      1 & 0 & 0 & -1 & 0 & 0 & z & z & 0 & 0 & 0 & 0 \\
      0 & -1 & 0 & 1 & 1 & 0 & 0 & 0 & 0 & 0 & 0 & 0 \\
      0 & 0 & 0 & 0 & -1 & 0 & 0 & z & 1 & 0 & 0 & 0 \\
      0 & 0 & 0 & 0 & 1 & -1 & 0 & 0 & 1 & 1 & 0 & 0 \\
      0 & 0 & 0 & 0 & 0 & 1 & -1 & 0 & 0 & 1 & 1 & 0 \\
      0 & 0 & 0 & 0 & 0 & 0 & 1 & -1 & 0 & 0 & 1 & 1 \\
    \end{array}
    \right)
    $
  }
  \centerline{}
  \centerline{
    $
    \det = (w^2 - w) z^4 + (-w^4 - 37w^3 - 137w^2 - 35w - 1)z^3 +
    (3w^4 - 175w^3 + 146w^2 - 2w)z^2$
  }
  \centerline{
        $+ (-3w^4-40w^3-w^2)z+w^4$
  }
  \caption{The dimer corresponding to the
    final zigzag path diagram in figure \ref{fig:example} and its Kasteleyn
    matrix.  The rows represent white nodes and the columns represent black
    nodes.}
  \label{fig:exampledimernumbered}
}
\end{FIGURE}

Figure \ref{fig:example} shows an example case of the algorithm.

\subsection{Proof of the algorithm}
We need to prove that we have the paths necessary to perform each step,
and that the finished dimer has properly ordered nodes and
has no extra crossings.

\begin{thm} \label{thm:algorithm}
At each step of the algorithm, the following are true:
\begin{enumerate}
\item If there are $m$ zigzag paths with winding $(a,b)$ and $n$
zigzag paths with winding $(-a,-b)$, then there are $\min(m,n)$ negatively
oriented pairs of $(a,b)$ and $(-a,-b)$ paths.  (This condition
ensures that we can always perform the next step of the algorithm.)
\item There are no extra crossings.
\item All nodes are properly ordered.
\end{enumerate}
\end{thm}

\begin{FIGURE} {\label{fig:pairtopair}
\centerline{\includegraphics[height=1.75in]{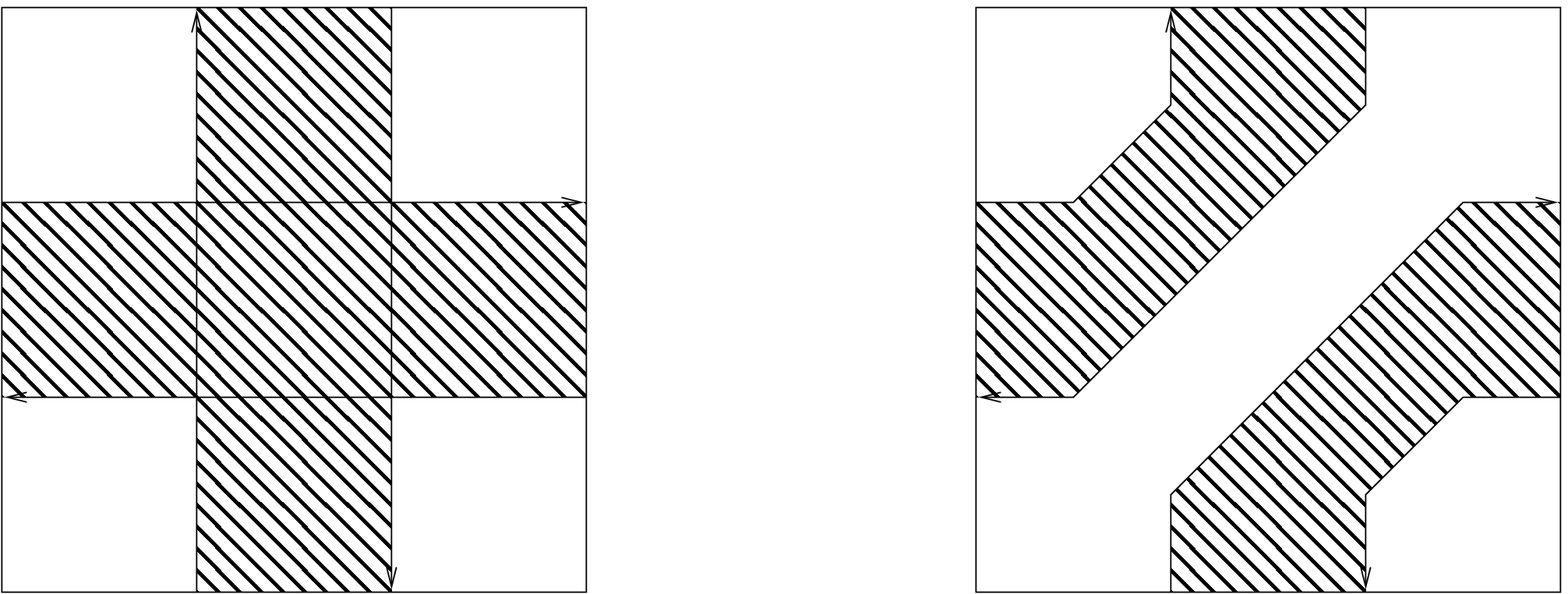}}
\caption{
Left: We start out with two negatively oriented pairs of opposite paths.
The shaded regions are free of crossings.  Right: The regions formed by
the merged pairs are still free of crossings.}
} \end{FIGURE}

\begin{proof}
It is clear that all of these conditions hold for the initial dimer.  Now let's
look at whether the first condition will be preserved.  Operation I will
preserve the condition for the winding of the paths being merged provided that
we merge unpaired paths when possible.  It will also satisfy the condition for
the windings of the newly created paths since there are no $(-a,-b)$ paths.
Operation II will preserve condition 1 for the windings of the paths being
merged since it only deletes negatively oriented pairs.  Figure
\ref{fig:pairtopair} illustrates
why Operation II creates negatively oriented pairs of opposite paths.
For Operation III we should again merge unpaired paths when possible.  It
is clear that the reinserted paths form pairs, and we can make these
pairs negatively oriented if we desire.

Now consider whether extra crossings are introduced.
Let the windings of the paths being merged be $(a,b)$, $(-a,-b)$,
$(c,d)$, and $(-c,-d)$, where $ad-bc=1$.  A path of winding $(e,f)$ will
have extra crossings with the new $(a+c,b+d)$ paths if
$af-be$ and $cf-de$ have opposite signs.  Equivalently, there will
be extra crossings if $f/e$ is between $b/a$ and $d/c$.  But because
of the Farey fraction ordering, there are no windings $(e,f)$ with this
property.  So extra crossings are not introduced.

Finally consider whether proper ordering is preserved.
Again let the windings of the paths being merged be $(a,b)$, $(-a,-b)$,
$(c,d)$, and $(-c,-d)$, $ad-bc=1$.  In Operation I, some nodes will see
an $(a,b)$ path or a $(c,d)$ path become an $(a+c,b+d)$ path.  Therefore
proper ordering is preserved provided there are no windings between
$(a,b)$ and $(c,d)$.  This is always the case because of the Farey fraction
ordering.  In Operation II, in addition to deletion we also need to move paths
past each other.  Some nodes are deleted and the others remain unchanged, so
proper ordering is preserved.  In Operation III, the process of making
the lone paths is the same as Operation I, so it preserves proper ordering.
Removing pairs also preserves proper ordering.
Inserting pairs of paths preserves proper ordering if each intersection between
a path in the pair and another path has the same sign as their crossing number,
i.~e.\ the paths in the pair do not have extra crossings.
Since we are inserting them along an existing path,
they will not have extra crossings if the existing path does not have any.
We have already showed that we never introduce extra crossings.
\end{proof}

\subsection{\label{sec:algextra}Allowing extra crossings}

If we want to produce diagrams with extra crossings, we can always 
just skip the steps for removing the extra crossings.  When we want to
create $(a,b)$ and $(-a,-b)$ paths, we just perform Operation I twice.
There is one
potential issue in that we have always assumed that the zigzag paths that we
join have just
one crossing.  We always join paths with oriented crossing number $\pm 1$,
but now the unoriented crossing number can be larger than the absolute
value of the oriented number.  But we recall that the only extra crossings
we create are between paths with windings of the form $(a,b)$ and $(-a,-b)$.
We may later merge these paths with some other paths, but
the extra crossings will always be between paths with oppositely signed
$x$-coordinates and oppositely signed $y$-coordinates.  We never merge such
pairs of paths.

\subsection{\label{sec:algunitarity}The number of independent solutions to
the $R$-charge equations}

We now exhibit the dimers required by Lemma \ref{lem:pathsminus1}.

\begin{lem} \label{lem:pm1construct}
The algorithm described in section \ref{sec:algextra} produces dimers for which
the set of all solutions to equations 
(\ref{eq:rvert}) and (\ref{eq:rface}) has dimension equal to the number of
zigzag paths minus one.
\end{lem}

\begin{proof}
Our proof is by induction.  Our algorithm starts with a dimer that is a
diamond-shaped grid.  We denote the position of an edge in the grid by $(i,j)$.
We can see (e. g. by Fourier analysis) that the general solution to
(\ref{eq:rvert}) and (\ref{eq:rface}) is $\frac{1}{2}+(-1)^i f(j) +(-1)^j g(i)$
for arbitrary functions $f,g$.  The number of independent solutions is the
number of rows plus the number of columns minus one (the minus one come from
the fact that $f(j)=(-1)^j, g(j)=-(-1)^i$ produces the same solution as
$f(j)=0, g(j)=0$), which is the
number of zigzag paths minus one.

Now consider what happens when our algorithm deletes an edge of the toric
diagram.  If we have a solution to the equations (\ref{eq:rvert}) and
(\ref{eq:rface}) in the
new dimer, we can construct a solution to the equations in the old dimer
by assigning a value of zero to the deleted edge.  Conversely, if we have
a solution in the old dimer in which the deleted edge has value zero, then
we have solution in the new dimer as well.  We know that there exists a solution
in the old dimer where the deleted edge is nonzero, since the deleted edge
is contained in some boundary perfect matching.  So deleting the edge reduces
the dimension of the solution space of (\ref{eq:rvert}) and (\ref{eq:rface})
by one, and also reduces the number of zigzag paths by one.

\end{proof}

\section{Conclusions}

We showed that dimers that have the number of faces predicted by the
AdS dual theory and that have valence one nodes will have
many nice properties: they are ``properly ordered'',
their cubic anomalies are in agreement with the Chern-Simons coefficients of
the AdS dual, gauge-invariant chiral primary operators satisfy
the unitarity bound, corner perfect matchings are unique,
and zigzag path windings are in one-to-one correspondence with the $(p,q)$-legs
of the toric diagram.

We derived some simple bounds for the cubic anomaly $a$ in terms of the
area of the toric diagram (and hence in terms of the number of gauge groups).

We provided a precise, computationally feasible algorithm for
producing a dimer model for a given toric diagram based on previous
partial resolution techniques and the Fast Inverse Algorithm.

It would be interesting to see if our results could apply to
the three-dimensional dimers discussed in \cite{Lee:2006hw} and
the orientifold dimers discussed in \cite{Franco:2007ii}.

\section*{Acknowledgments}
I would like to thank Christopher Herzog for suggesting this problem to
me and for providing many helpful discussions.  I would like to thank
Daniel Kane for helpful discussions regarding Lemmas \ref{lem:pathsminus1} and
\ref{lem:pm1construct} and Alberto Zaffaroni for discussion.
I would also like to thank Amihay Hanany for introducing me to dimer models.
This work was supported in part by the NSF Graduate Fellowship Program and
NSF Grant PHY-075696.

\bibliography{algorithm}{}
\bibliographystyle{JHEP}
\end{document}